\documentclass[journal]{IEEEtran}

\usepackage{setspace}
\usepackage{lipsum}
\usepackage{mathtools}
\usepackage{cuted}
\usepackage{amssymb}
\usepackage{amsfonts}
\usepackage{epsfig}
\usepackage{calc}
\usepackage{color}
\usepackage[all]{xy}
\usepackage{bm}
\usepackage{enumerate}
\usepackage{amsmath}
\usepackage[hyperindex,colorlinks=true]{hyperref}
\usepackage{graphicx}
\DeclareGraphicsExtensions{.pdf,.png,.jpg}
\usepackage{caption}
\usepackage{subcaption}

\newtheorem{thm}{Theorem}[section]
\newtheorem{cor}[thm]{Corollary}

\newtheorem{lem}[thm]{Lemma}

\newtheorem{note}{Remark}

\newcommand{\bit}{\begin{itemize}}
\newcommand{\eit}{\end{itemize}}
\newcommand{\bcor}{\begin{cor}}
	\newcommand{\ecor}{\end{cor}}

\begin{document}

\title{Updating Content in Cache-Aided Coded Multicast}
\author{Milad Mahdian, N. Prakash, Muriel M{\'{e}}dard and Edmund Yeh
	\thanks{Milad Mahdian is an Associate in Technology at Goldman Sachs Group, Inc. N. Prakash, and Muriel M{\'{e}}dard are with the Research Laboratory of Electronics, Massachusetts Institute of Technology (email: \{prakashn, medard\}@mit.edu).  Edmund Yeh is with Department of Electrical and Computer Engineering, Northeastern University (email: eyeh@ece.neu.edu). The work of N. Prakash is in part supported by the AFOSR award  FA9550-14-1-0403. Edmund Yeh gratefully acknowledges support from National Science Foundation Grant CNS-1423250 and a Cisco Systems Research Grant.}}	

\maketitle

\begin{abstract}
	Motivated by applications to delivery of dynamically updated, but correlated data in settings such as content distribution networks, and distributed file sharing systems, we study a single source multiple destination network coded multicast problem in a cache-aided network. We focus on models where the caches are primarily located near the destinations, and where the source has no cache. The source observes a sequence of correlated frames, and is expected to do frame-by-frame encoding with no access to prior frames. We present a novel scheme that shows how the caches can be advantageously used to decrease the overall cost of multicast, even though the source encodes without access to past data. Our cache design and update scheme works with any choice of network code designed for a corresponding cache-less network, is largely decentralized, and works for an arbitrary network. We study a convex relation of the optimization problem that results form the overall cost function. The results of the optimization problem determines the rate allocation and caching strategies. Numerous simulation  results are presented to substantiate the theory developed.
\end{abstract}

\section{Introduction}

In this paper, we study a single source multiple destination network coded multicast problem in a cache-aided network, where the source observes a sequence of correlated frames. The source encodes each frame individually and only has the access to the current frame, and knowledge of the correlation model while encoding. The source does not have any cache to store old frames. We consider transmission over an arbitrary network, where  a certain subset of nodes (excluding the source) has access to local caches with a certain cost model. The cache is simply a local storage box that can be used to store a part of previously received packets, and is used while generating the output coded packets for the future rounds. Each node has the flexibility to update its cache after each round. We are interested in devising strategies for source encoding, intermediate node processing including caching, and destination nodes' processing, including decoding and caching such that the overall communication and caching costs are optimized. 

The setting, in one direction, is a natural generalization of the standard coded multicast problem studied in ~\cite{Lun}, with the main difference being the presence of caches in a subset of the nodes. Instantaneous encoding (frame by frame) is an important requirement in delay sensitive applications such as video conferencing, online gaming. Network coding is widely recognized as  a superior alternative to routing in delay sensitive applications~\cite{chen2017delay}, \cite{chou2003practical}, \cite{gkantsidis2005network}, \cite{llorca2013network}. Network coding \cite{fragouli2007network}, \cite{li2003linear} specifically, random linear network coding~\cite{rlnc} avoids the need for dynamically constructing multicast trees based on latency tables as needed for multicast routing, and permits a decentralized mode of operation making optimal use of network resources. 

Caching for content distribution in networks is an active area of research with several works studying the cache placement and delivery protocols~\cite{maddah2014fundamental, maddah2015decentralized, niesen2017coded, park2017coded, hassanzadeh2016cache, yang2017centralized}. In the traditional caching settings, users are often interested in a specific subset of the popular files. With limited cache sizes, the question of interest is what to store in the various caches such that as many users' requests can be served while minimizing download from the backend (a distant cloud). We note that in this work, we are exploring a slightly different problem that arises in the context of caching. Firstly, we restrict ourself to multicast networks, in which case, the problem of cache placements as studied in the above works becomes trivial (since all destinations are interested in the same message). Secondly, our  focus is on identifying efficient schemes for updating the content of the distant caches over a network when there is limited cache available at the source itself. We also note that the above works pay little attention to the cache placement cost or updating the caches once content changes, which is the key focus of this work. Furthermore, in the above works, the network topologies considered are simplistic, like  a direct shared link from the edge caches to the backend or tree topology. In this context, our problem can be considered as one of efficient cache placement for multicasting in arbitrary networks. 

We note that presence of a cache at the source means that the source can automatically compute the differences across frames, and only need to transmit the difference frames. In our model, the assumption of lack of cache at the source is motivated more by a practical point of view rather than the need to construct a theoretical problem. This is because, in content distribution network, it is easy and efficient to place caches closer to the destinations than deeper within the network.  

Our model is also applicable in distributed file storage systems shared by several users who want to update 
the stored content. Distributed storage systems like those used in Apache Cassandra, Amazon Dynamo DB, often replicate every user file across several servers in the network, and allow these files to be updated by one of several users. For updating a file, any of the users contact a proxy client at the edge of the network, and this proxy client then takes the job of multicasting the user file to the target replica servers. When compared to our setting, the replica servers act as the destinations, and have caches. The proxy client acts as the source. The sequences of frames multicast by the source correspond to the various file versions updated by various users. Systems like Apache Cassandra, Amazon Dynamo DB handle several hundred thousands of files, any proxy client is used for updating several of these files. It is inconceivable from a storage cost point of view for the proxy client at the edge to cache the latest version of every file that it ever sees. In this context, the proxy client is best modeled via a source which only has access to the current frame that it wants to multicast.

In this work, we present a novel scheme for the cache-aided multicast problem that works with any choice of network code designed for a corresponding network without any caches. Given the network code for the corresponding network without caches, our solution for cache design is largely decentralized. The only coordination needed is for any node to inform its incoming neighbors about the coding coefficients, and whether caching is used or not. Our solution works for an arbitrary multicast network, and is capable of utilizing presence of caches at any of the nodes in the network. Finally, our solution for cache design is one that permits an overall cost optimization study as in \cite{Lun}. The result of the optimization determines rate allocation for each link, and also whether a cache is to be used or not (assuming it is available).

In the rest of the introduction, we present our system model, summary of contributions and other related works.

\subsection{System Model} \label{sec:sys}

We consider the problem of multicasting in a directed acyclic network having a single source and $L, L \geq 1$ destinations. We let $\mathcal{N} = \{V_1, \ldots, V_N\}$ to denote the nodes in the network.  Without loss of generality, we let $V_1$ to denote the source node, and $V_{N-L+1}, \ldots, V_N$ to denote the $L$ destination nodes. { The remaining nodes $\{V_2, \ldots, V_{N-L}\}$ are intermediary nodes in the network. The source node generates data, the destinations consume data, and the intermediary nodes simply transmit, to the various outgoing links, an appropriate function of data received thus far. Example networks appear in Fig. \ref{fig:topos}; for instance in Fig. \ref{topo:butterfly}, there are two destinations, and four intermediary nodes in addition to the source node.}  The source node shall also be denoted as $S$, and the destination nodes by $D_1, \ldots, D_L$ with $V_{N - L + i} = D_i, 1 \leq i \leq L$. We assume node $V_i$ to have ${\alpha}_i$ incoming edges and $\beta_i$ outgoing edges.  For any node $V_i \in \mathcal{N}$, we refer to $i$ as the index of the node. For any node $\mathcal{V}_i$, we write $\mathcal{N}_{in}(i)$ and $\mathcal{N}_{out}(i)$ to respectively denote the indices of nodes corresponding to the incoming and outgoing edges of $V_i$. The source has no incoming edge, and the destinations have no outgoing edge. We assume that there is at most one edge between any two nodes. We use $\mathcal{E}$ to denote the set of all edges in the network. Any element of $\mathcal{E}$ is of the form $(i, j)$, and will indicate the presence of a directed edge from $V_i$ to $V_j$. We also write $V_i \rightarrow V_j$ to indicate an edge $(i, j) \in \mathcal{E}$.

The source is expected to multicast a sequence of $M$ frames ${\bf m}^{(1)}, \ldots, {\bf m}^{(M)}, {\bf m}^{(i)} \in \mathbb{F}_q^B$ to the $L$ destinations in $M$ rounds. { Here $\mathbb{F}_q$ denotes the finite field of $q$ elements, and $\mathbb{F}_q^B$ denotes the $B$-dimensional vector space over $\mathbb{F}_q$. The elements of $\mathbb{F}_q^B$ are assumed to be column vectors.} In round $i, 1\leq i \leq M$, the source has access to (only) ${\bf m}^{(i)}$ during the encoding process. The frames are correlated; specifically, we assume that any two adjacent frames differ in at most $\epsilon$ symbols, i.e., Hamming weight$({\bf m}^{(i+1)} - {\bf m}^{(i)}) \leq \epsilon, 1 \leq i \leq M-1$.  We are interested in zero-error decoding at all the $L$ destinations. Further, we assume a worst-case scenario model, where the scheme must work for every possible sequence of the input frames that respect the correlation model described above.

Each node, except the source node, has access to local storage that can be used to cache (coded) content from previous rounds. The data encoded by a node $V \in \mathcal{N}$ during round $i$ is a function of its incoming data during the $i^{\text{th}}$ round, and the cache of node $V$ after round $i-1$. We assume that all caches to be empty, initially. The cache content of node $V$ after round $i$ is also possibly updated; if so, this is a function of its incoming data during the $i^{\text{th}}$ round and the cache of node $N$ after round $i-1$. We associate the cost function $f_i:\mathbb{N}\rightarrow \mathbb{R}$ with each cache of node $V_i \in \mathcal{N}$,  where $f_i(s), s \geq 1$ denotes the cost to store $s$ bits for one round. Since we calculate cost per bits cached, the cost function definition is not tied to the specific finite field $\mathbb{F}_q$ that is used to define the source alphabet.

We further associate the cost function $f_{i,j}:\mathbb{N}\rightarrow \mathbb{R}$ with the directed edge from node $V_i$ to node $V_j$. The quantity $f_{i,j}(s), s\geq 1$ shall denote the cost to transmit, in a single round, a packet of size $s$ bits along the directed edge $V_i \rightarrow V_j$.  If the directed edge $V_i \rightarrow V_j$ does not exist in the network, we simply assume $f_{i,j}(s) = \infty, s\geq 1$.

\subsection{Summary of Contributions} \label{sec:problem}

We are interested in devising strategies for source encoding, intermediate node processing including caching, and destination node processing, including decoding and caching such that the overall communication and storage costs are optimized. 

We present a novel achievable scheme for the above problem that takes advantage of the correlation among the source frames. We show that even though the source does not have any local storage, it is still possible to encode at the source and the various intermediate nodes in a way that takes advantage of the local storage capabilities at the various intermediate nodes. In the special case where the cost due to storage is negligible when compared to cost due to communication, the overall cost - which is simply the cost due to communication - is significantly better than any scheme that does not utilize the caches.
The overall cost computation is posed as an optimization problem such that any solution to this optimization problem can be mapped to an instance of the achievable scheme that we present here. We also present explicit and numerical solutions to the optimization problem for three different network topologies in order to demonstrate the benefits of our scheme.

In our achievable scheme, cache of any node is used with an aim to decrease the communication cost on the incoming links on that node.  Our achievable scheme relies on the work in \cite{PraMed}, where the authors present a coding scheme for updating linear functions. We review the relevant results in the next section, before presenting our achievable scheme  in Section \ref{sec:scheme}. We also discuss two straightforward variations of the setting in \cite{PraMed}, including new achievability results for the variations, that are needed for constructing an achievable scheme for our problem. 
\vspace{0.1in}

\begin{note}
	An approach where cache of a node is used to decrease the communication cost on the outgoing links is also conceivable; for example if the source itself has unlimited cache, {the source can simply cache ${\bf m}^{(i)}$ during round $i$, and then multicast the difference vector ${\bf m}^{(i+1)} - {\bf m}^{(i)}$  in round $(i+1)$. This simple solution works as long as the all the destinations have the ability to cache ${\bf m}^{(i)}$ during round $i$. } Further, even if the source does not have any cache like in our model, the source can divide the frame to be sent to several sub-frames and perform multiple-unicast transmissions to a bunch of nodes that have caches whose total cache size exceeds the size of the frame. Each of these nodes, having access to the sub-frame from the past rounds, computes the difference for the sub-frame and multicasts the sub-frame difference to the destinations. While the strategy is reasonable in situations when there are a few large-sized-cache nodes near the source, and caches at the destinations, its merits are not immediately clear when the cache sizes are only a fraction of the frame size $B$, and if the nodes with such caches are distributed through out the network. A study of approaches along this direction is left for future work. Furthermore, in this work, our focus is on networks where the communication cost of the incoming links of a node is positively correlated with the size of its cache. In such networks, it is intuitive to consider schemes (as discussed in this work) where the cache of a node is used to minimize the communication cost on the incoming links of the node. 
\end{note}

\subsection{Other Related Work}\label{sec:related}

\subsubsection{Network Coding for Multicasting, Optimization Problems}

Cost minimization for multicast, either based on routing or network coding, involves assigning costs to each links and minimizing overall cost of multicast. These costs usually characterize bandwidth costs.  For example, see works \cite{Lun}. The cost can also represent delay of the link, and this can be a secondary cost metric in addition to the primary bandwidth cost metric.  For example, while formulating the optimization problem, one can impose a QoS requirement by restricting the maximum number of hops in a path from source to destination, and thus addressing the delay requirement. Works that address delay requirements for multicast appear in ~\cite{chen2017delay, cheng2017delay}.The authors of \cite{chen2017delay} propose a delay sensitive  multicast protocol for MANETs. The protocol is based on routing packets, and does not use either  network coding or caching. In \cite{cheng2017delay}, the authors propose new methods for monitoring and updating latency information of links for delay sensitive multicast routing.

\subsubsection{Caching for Serving Popular Demand}

As noted above there are several works in the literature that study the problem of cache placement for serving popular data~\cite{maddah2014fundamental, maddah2015decentralized, niesen2017coded}. It is an interesting problem to explore how the solution in this work can be used to update caches when using schemes are in ~\cite{maddah2014fundamental, maddah2015decentralized, niesen2017coded} for non-multicast demands. The setting of \cite{park2017coded} is related in which the authors study a setting where there are edge caches, users requests are served based on the content of the caches and a coded multi-cast packet that is obtained from the backend. The back-end is assumed to connected to the various destinations via a network. Thus while emphasis is placed on minimizing delivery cost over an arbitrary network via network coding, the cost of cache placement is ignored, and also the setting does not consider the problem of updating caches.

Recently, there have been efforts to generalize the setting of edge caching to the setting where there is correlation among the various files~\cite{hassanzadeh2016cache, yang2017centralized}. The traditional setup in \cite{maddah2014fundamental} assumes that the various files at the distant source are independent. In the generalized setting, a correlation model is assumed to capture file dependencies. Both placement and the delivery phases are designed accounting for this correlation. We note that this generalization of the caching problem is one step closer to our model than the model of  \cite{maddah2014fundamental}. In our model, we also assume  correlation among the various packets that need to be multicast to the destinations. It is an interesting future problem to see if our technique in this work can be used to provide alternate solutions to the delivery phase of a setting as in  \cite{hassanzadeh2016cache}, perhaps by restricting to a correlation model as in this work.

The rest of the document is organized as follows. In Section \ref{sec:updates}, we review the necessary results from \cite{PraMed} for function updates for the point-to-point setting. We also generalize some of these results towards constructing a scheme for our problem. Our scheme for the multicast problem will be then presented in Section \ref{sec:scheme}. The expression for the overall cost associated with the scheme will be obtained in Section \ref{sec:opt}, and posed an optimization problem. A convex relaxation is then presented, whose solution is obtained via primal-dual algorithms. We also show how to relate solutions to the original problem from the solutions of the relaxed problem. In Section \ref{sec:exp}, we present simulation results of our scheme for three networks namely the butterfly network, service network and the CDN  network. For each of these networks, we show how caching at or near the edge improves the cost significantly when compared to the case of no caching. We  present results for various caching cost metrics like no caching cost, linear or quadratic caching costs. 

\section{Coding for Function Updates} \label{sec:updates}

In \cite{PraMed}, the authors consider a point-to-point communication problem in which an update of the source message needs to be communicated to a destination that is interested in maintaining a linear function of the actual source message. In their problem, the updates are sparse in nature, and where neither the source nor receivers are aware of the difference-vector.  The work proposes encoding and decoding schemes that allow the receiver to update itself to the desired function of the new source message, such that the communication cost is minimized.  

The ``coding for function updates" model used in \cite{PraMed} that is relevant to the work here, is shown in Fig. \ref{fig:coding_updates}. The $B$-length column-vector ${\bf x} \in \mathbb{F}_q^B$ denotes the initial source message. The receiver maintains the linear function $A{\bf x}$, where $A$ is a $\theta \times B$ matrix, $\theta \leq B$, over $\mathbb{F}_q$. The updated source message is given by ${\bf x} + {\bf e}$, where ${\bf e}$ denotes the difference-vector, and is such that $\text{Hamming wt.}({\bf e}) \leq \epsilon$, $0 \leq \epsilon \leq B$. Encoding at the source is linear and is carried out using the $\gamma \times B$ matrix $H$.  The source is aware of the function $A$ and the parameter $\epsilon$, but does not know the vector ${\bf e}$. Encoding at the source is such that the receiver can update itself to $A({\bf x}  + {\bf e})$, given the source encoding and $A{\bf x}$.  Assuming that the parameters $n, q, \epsilon$ and $A$ are fixed, the communication cost of this problem is defined as the parameter $\gamma$, which is the number of symbols generated by the linear encoder $H$. The authors in \cite{PraMed} assume a zero-probability-of-error, worst-case-scenario model; in other words, the scheme allows error-free decoding for every ${\bf x} \in \mathbb{F}_q^B$ and $\epsilon$-sparse ${\bf e} \in \mathbb{F}_q^B$. 

\begin{figure}[h]
	\centering
	\includegraphics[width=3in]{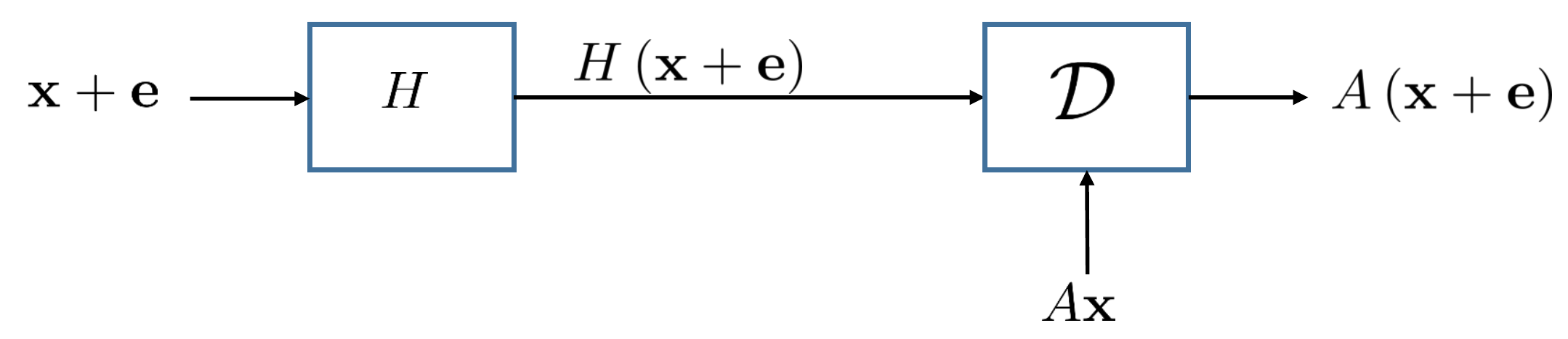}
	\caption{The model used in \cite{PraMed} for linear function updates.}
	\label{fig:coding_updates}
\end{figure}

The scheme in \cite{PraMed} is one which has optimal communication cost $\gamma$. The authors show an achievability and a converse for the above setting. In our problem, we are only interested in the achievability part of their result. We state this explicitly in the following lemma. 

\vspace{0.1in}

\begin{lem}[\cite{PraMed}] \label{lem:func_updates}
For the function-updates problem in Fig.\ref{fig:coding_updates}, there exists a linear encoder $H$ of the form $H = SA, S \in \mathbb{F}_q^{\gamma \times \theta}$, and a decoder $\mathcal{D}$ such that the output of the decoder is $A({\bf x}  + {\bf e})$ for all ${\bf x}, {\bf e} \in \mathbb{F}_q^B$, Hamming wt.$({\bf e}) \leq \epsilon$, whenever $q \geq 2\epsilon B^{2\epsilon}$ and $\gamma \geq \min(2\epsilon, \text{rank}(A))$. 
\end{lem}
\begin{proof}
See Section III, \cite{PraMed}.	
\end{proof}

In the above lemma, the quantity $\mathbb{F}_q^{\gamma \times \theta}$ denotes the set of all $\gamma \times \theta$ matrices whose elements are drawn from $\mathbb{F}_q$.

\subsection{Variations of the Function Updates Problem Needed in this Work}

We now present two variations of the function-updates problem (see Fig.\ref{fig:coding_updates}) that we need to construct an achievable scheme for our problem in this paper. For both variations, we are interested in achievable schemes, and these follow directly from Lemma \ref{lem:func_updates}. 

In the first variation (see Fig. \ref{fig:coding_updates_var1}), the source input is changed to $P({\bf x}  + {\bf e})$ (instead of {\bf x}  + {\bf e}), where $P$ is a $\theta' \times B$ matrix such that $A = CP$ for some $\theta \times \theta'$ matrix  $C$. The following lemma guarantees an achievable scheme for this variation.

\vspace{0.1in}

\begin{lem}[\cite{PraMed}] \label{lem:func_updates_var1}
	For the function-updates problem in Fig.\ref{fig:coding_updates_var1}, there exists a linear encoder $H$ of the form $H = SC, S \in \mathbb{F}_q^{\gamma \times \theta}$, and a decoder $\mathcal{D}$ such that the output of the decoder is $A({\bf x}  + {\bf e})$ for all ${\bf x}, {\bf e} \in \mathbb{F}_q^B$, Hamming wt.$({\bf e}) \leq \epsilon$, whenever $q \geq 2\epsilon B^{2\epsilon}$ and $\gamma \geq \min(2\epsilon, \text{rank}(A))$. 
\end{lem}
\begin{proof}
If the input had been ${\bf x}  + {\bf e}$, instead of $P({\bf x}  + {\bf e})$, then we know from Lemma \ref{lem:func_updates} that there exists a { linear encoder $H'$ of the form} $H' = SA, S \in \mathbb{F}_q^{\gamma \times \theta}$, and a decoder $\mathcal{D}'$ such that the output of the decoder is $A({\bf x}  + {\bf e})$ for all ${\bf x}, {\bf e} \in \mathbb{F}_q^B$, Hamming wt.$({\bf e}) \leq \epsilon$, whenever $q \geq 2\epsilon n^{2\epsilon}$ and $\gamma \geq \min(2\epsilon, \theta)$. Since $A = CP$, $H' = SCP$, and thus { $H'({\bf x}  + {\bf e}) = H(P{\bf x}  + P{\bf e})$}, where the candidate encoder $H$ in Fig. \ref{fig:coding_updates_var1} is chosen as $H = SC$. The candidate decoder $\mathcal{D} = \mathcal{D}'$.
\end{proof}

\vspace{0.1in}

\begin{note}
One can convert the problem in Fig. \ref{fig:coding_updates_var1} to the same form as in Fig. \ref{fig:coding_updates}, by letting ${\bf x}' = P{\bf x}$ and ${\bf e}' = P{\bf e}$ in Fig. \ref{fig:coding_updates_var1}, so that the input becomes ${\bf x}' + {\bf e}'$ and the decoder side-information becomes $C{\bf x}'$. However, there is no natural relation between the sparsity of ${\bf e}$ and ${\bf e}'$; in fact we only know that Hamming wt.$({\bf e}') \leq \theta'$. As a result, a direct application of Lemma \ref{lem:func_updates} to the converted problem only guarantees an achievability scheme whose communication cost is given by 
\begin{eqnarray}
\gamma' & = & \min(2\theta', \text{rank}(C)) \\
& = & \text{rank}(C) \\
& \geq & \text{rank}(A) \\
& \geq & \min(2\epsilon, \text{rank}(A)) = \gamma,
\end{eqnarray}
where $\gamma$ is the communication cost guaranteed by Lemma \ref{lem:func_updates_var1}. As a result, we will use the scheme guaranteed by Lemma \ref{lem:func_updates_var1} instead of the obvious conversion discussed here.
\end{note}

\begin{figure}[h]
	\centering
	\includegraphics[width=3in]{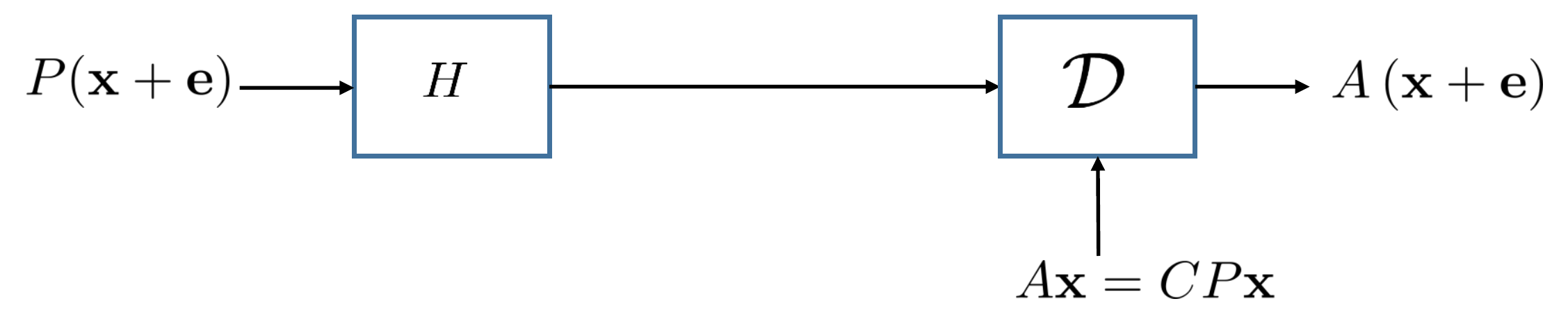}
	\caption{A { variant } of the function updates problem shown in Fig. \ref*{fig:coding_updates}.}
	\label{fig:coding_updates_var1}
\end{figure}

The second variation that we need is a multi-source variation of the model in Fig. \ref{fig:coding_updates_var1} (see Fig. \ref{fig:coding_updates_var2}). In this model, there are $\alpha$ sources, $\alpha \geq 1$. The $i^{\text{th}}, 1 \leq i \leq \alpha$ source input is given by $P_i({\bf x}  + {\bf e})$ where $P_i$ is a $\theta'_i \times B$ matrix. The decoder side-information is $A = CP$, where $C = [ C_1 \ C_2 \ \ldots C_{\alpha}], C_i \in \mathbb{F}_q^{\theta \times \theta'_i}$ and 
\begin{eqnarray}
P & = & \left[ \begin{array}{c} P_1 \\ P_2 \\ \vdots \\ P_{\alpha} \end{array} \right],  \ P_i \in 	\mathbb{F}_q^{\theta'_i \times B}.
\end{eqnarray}	
The following lemma guarantees an achievable scheme for this variation.

\vspace{0.1in}

\begin{lem}[\cite{PraMed}] \label{lem:func_updates_var2}
	For the function-updates problem in Fig.\ref{fig:coding_updates_var1}, there exists encoders $\{H_i, 1 \leq i \alpha\}$ of the form  $H_i = SC_i, S \in \mathbb{F}_q^{\gamma \times \theta}$, and a decoder $\mathcal{D}$ such that the output of the decoder is $A({\bf x}  + {\bf e})$ for all ${\bf x}, {\bf e} \in \mathbb{F}_q^B$, Hamming wt.$({\bf e}) \leq \epsilon$, whenever $q \geq 2\epsilon B^{2\epsilon}$ and $\gamma \geq \min(2\epsilon, \text{rank}(A))$. 
\end{lem}
\begin{proof}
The matrix $S$ is chosen using Lemma \ref{lem:func_updates_var1} as though we are designing an encoder for the single source problem whose input is $P({\bf x}  + {\bf e})$. The decoder as a first step uses all the inputs to compute 
\begin{eqnarray}
\sum_{i = 1}^{\alpha}H_i\left(P_i({\bf x}  + {\bf e}) \right) & = & S \sum_{i = 1}^{\alpha}C_iP_i({\bf x}  + {\bf e})\\
& = & SCP({\bf x}  + {\bf e}),
\end{eqnarray}
and as a second step uses the decoder that is guaranteed by Lemma \ref{lem:func_updates_var1} for the  single source problem whose input is $P({\bf x}  + {\bf e})$. It is clear that the such a scheme works for the problem in Fig. \ref*{fig:coding_updates_var2}.
\end{proof}
	
\begin{figure}[h]
	\centering
	\includegraphics[width=3in]{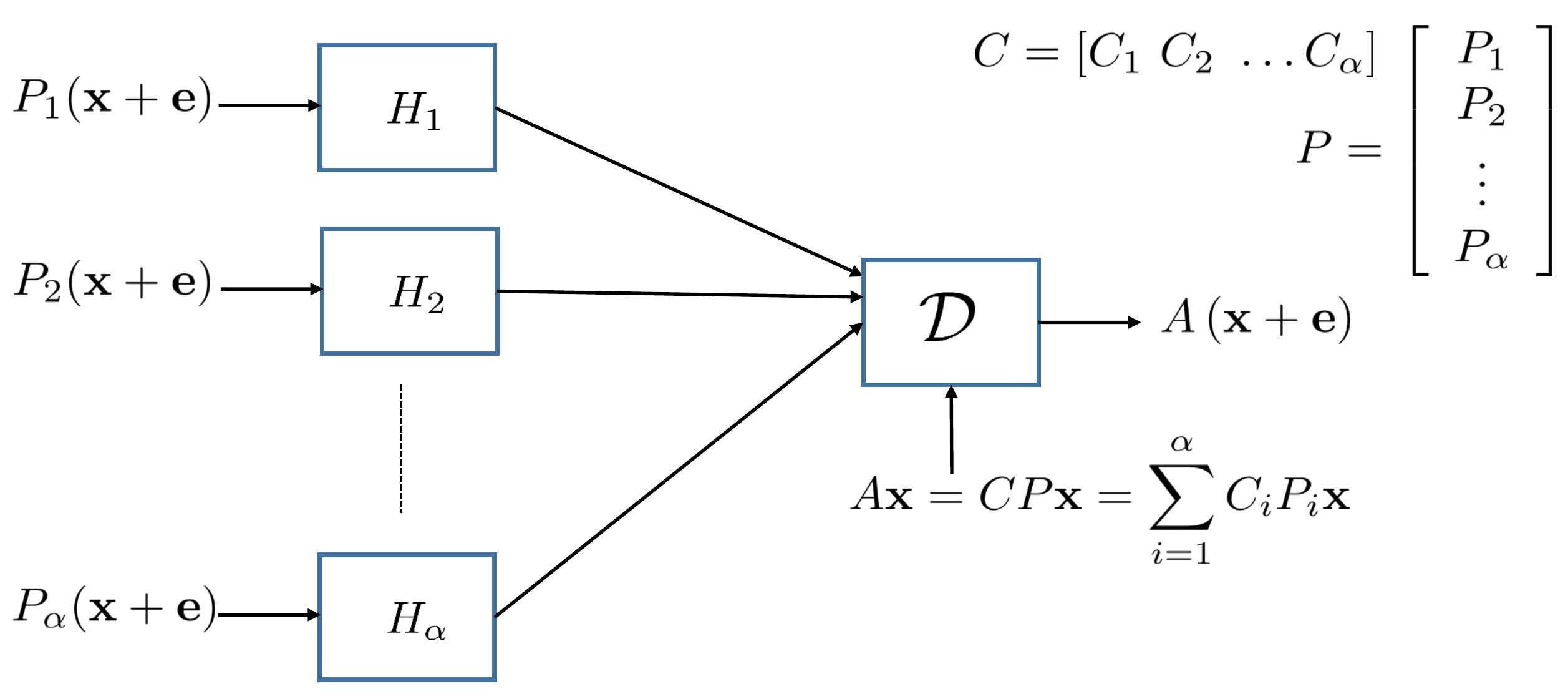}
	\caption{A multi-source { variant} of the function updates problem shown in Fig. \ref*{fig:coding_updates_var1}.}
	\label{fig:coding_updates_var2}
\end{figure}

\section{Our Scheme} \label{sec:scheme}

We are now ready to describe an achievable scheme for problem described in Section \ref{sec:problem}. We describe strategies for source encoding, intermediate node processing including caching, and also the processing at the destinations. The first round does not benefit from the caches, while the remaining rounds can potentially benefit from them. Below, we specify the strategies separately for the first round, and for any round beyond the first.  We shall use $\mathcal{S}$ to denote the achievable scheme presented here. The scheme is designed for fixed values of $B, \epsilon, M$ and for a fixed directed acylic graph whose nodes and edge set are given by $\mathcal{N}$ and $\mathcal{E}$. 

\subsection{Round $1$}

During round $1$, caches are empty and not useful. The problem is simply one of multicasting the frame $F_1$ from $S$ to $D_1, \ldots D_L$ on a directed acyclic network. We use a linear network code (LNC) $\mathcal{C} = \{G_1, \ldots, G_N\}$ for multicasting $F_1$ in round $1$, where $G_i$ denotes the matrix of coding coefficients at node $V_i, 1 \leq i \leq N$. 
We shall use ${\bf x}_i^{(1)}$ and ${\bf y}_i^{(1)}$ to denote the input and output vectors, respectively, at node $V_i$ during round $1$, where 
{ 
\begin{eqnarray}
{\bf x}_i^{(1)} & = & \left[ \begin{array}{c} {\bf x}_{i, e_1}^{(1)} \\ {\bf x}_{i, e_2}^{(1)} \\ \vdots \\ {\bf x}_{i, e_{\alpha_i}}^{(1)}\end{array}  \right], \begin{array}{c} \{e_1, e_2, \ldots, e_{\alpha_i}\}  = \mathcal{N}_{in}(i) \\ i \neq 1 \end{array} \label{eq:nodeinput} \\
{\bf y}_i^{(1)} & = & \left[ \begin{array}{c} {\bf y}_{i, f_1}^{(1)} \\ {\bf y}_{i, f_2}^{(1)} \\ \vdots \\ {\bf y}_{i, f_{\beta_i}}^{(1)}\end{array}  \right], \begin{array}{c}  \{f_1, f_2, \ldots, f_{\beta_i}\}  = \mathcal{N}_{out}(i) \\ i \notin \{N-L+1, \ldots, N\}\end{array}.  \label{eq:xy}
\end{eqnarray}
}
In the above equation, the column vector ${\bf x}_{i, e_{\ell}}^{(1)}, e_{\ell} \in \mathcal{N}_{in}(i)$ denotes the input received from the in-neighbor $V_{e_{\ell}}$  and  ${\bf y}_{i, f_{\ell}}^{(1)}, f_{\ell} \in \mathcal{N}_{out}(i)$ denotes the output sent to the out-neighbor $V_{f_{\ell}}$. Under this notation, we have 
\begin{eqnarray}
{\bf x}_{i, j}^{(1)} & = & {\bf y}_{j, i}^{(1)},  1 \leq i \leq N,
\end{eqnarray}
whenever $j \in \mathcal{N}_{in}(i)$ and $i \in \mathcal{N}_{out}(j)$. We assume that the LNC $\mathcal{C}$ is such that the quantities  ${\bf x}_i^{(1)}$ and ${\bf y}_i^{(1)}$ are related as 
\begin{eqnarray} \label{eq:round1_LNC}
{\bf y}_i^{(1)} & = & G_i {\bf x}_i^{(1)}, 1 \leq i \leq N.
\end{eqnarray}
We assume that $G_i$ has a form given by 
\begin{eqnarray}
G_i & = [ G_{i, 1} \ G_{i, 2} \ \ldots, G_{i, \alpha_i}],
\end{eqnarray}
such that ${\bf y}_i^{(1)} = \sum_{\ell = 1}^{\alpha_i}G_{i, \ell}{\bf x}_{i,e_{\ell}} ^{(1)}$. { Note that \eqref{eq:nodeinput} and \eqref{eq:xy} do not apply at the source and the destination(s), respectively.}  The quantity ${\bf x}_1^{(1)} = {\bf m}^{(1)}$, the input frame for round $1$. Further, the quantity ${\bf y}_{N - L + i}^{(1)}, 1 \leq i \leq L$ denotes the decoded output at the destination node $D_i$. The LNC $\mathcal{C}$ is chosen such that
\begin{eqnarray}
{\bf y}_{N - L + i}^{(1)} & = & {\bf m}^{(1)}, 1 \leq i \leq L. \label{eq:decode_rd1}
\end{eqnarray}
Note that such an LNC always exists since, a trivial solution for the LNC $\mathcal{C}$ is to simply pick ${\bf y}_{i, j}^{(1)}  = {\bf x}_i^{(1)}, 1 \leq j \leq \beta_i, 1 \leq  i \leq N$. Of course, as we shall see, our goal to pick the LNC (along with the strategy for the other rounds) so as to optimize the cost of communication and storage for the entire $M$ rounds. 

As for the caching policy during round $1$, we assume that each node $V_i , i > 1$  either caches the entire output vector ${\bf y}_{i}^{(1)}$ or does not cache anything\footnote{ 
{ In this work, we focus on a single connection form a single source. In practice, it is likely that a certain intermediate node, or even the destination node is simultaneously involved in multiple connections. In this scenario, one strategy is to divide the available cache at a node to a whole number of simultaneous connections, such that each connection either gets all the necessary cache or gets nothing. }}. We use the indicator variable $\delta_i$ to determine if node $V_i$ caches or not; $\delta_i = 1$ if $V_i$ caches; else $\delta_i = 0$.  We note that the decision to cache or not is be chosen (see Section \ref{sec:opt}) so as to optimize the cost of communication and storage for the entire $M$ rounds.

\subsection{Rounds $2, \ldots, M$}

The basic idea is to continue using the same LNC $\mathcal{C}$ for every round after the first as well, but we take advantage of the cache content, if any, to decrease the communication requirement. To implement the code $\mathcal{C}$ at any node $V_i \in \mathcal{N}$ in round $r, 2 \leq r \leq M$, we perform encoding at the various in-neighbors of $V_i$ such that given the input vector and the cache content (if present), node $V_i$ manages to generate the output vector ${\bf y}_i^{(r)}$ given by 
\begin{eqnarray}
{\bf y}_i^{(r)}& = & G_i{\bf x}_i^{(r)}, \label{eq:hypothetical}
\end{eqnarray}
{ where the quantity ${\bf y}_i^{(r)}$ is defined analogously as in \eqref{eq:xy} (by simply replacing the superscript  $(1)$ with $(r)$). The quantity ${\bf x}_i^{(r)}$ is the input that node $i$ would receive in the absence of cache from previous round\footnote{{The quantity ${\bf x}_i^{(r)}$ does not represent the actual input to node $i$, if node $i$ employs a cache. We avoid a general notation for input to node $i$ for round $r$, since this is not needed for the discussion.}}, and is defined like in \eqref{eq:nodeinput}. If there is no cache, we further know that 
\begin{eqnarray}
{\bf x}_{i, e_j}^{(r)} & = & {\bf y}_{e_j, i}^{(r)},  e_j \in \mathcal{N}(i).
\end{eqnarray} 
When there is a cache at node $i$, node $e_j, e_j \in \mathcal{N}(i)$ instead of sending ${\bf y}_{e_j, i}^{(r)}$ as it is,  encodes ${\bf y}_{e_j, i}^{(r)}$ using the scheme described  by Lemma \ref{lem:func_updates_var2}, so that we have a lower the communication requirement, between the in-neighbors of $V_i$ and $V_j$, than what we need during round $1$ (so as to satisfy \eqref{eq:hypothetical} in round $r$). Before describing the usage of  Lemma \ref{lem:func_updates_var2},  we note that the caching policy for any round is same as that of round $1$, i.e., if $\delta_i = 1$ after the first round, the node simply replaces existing cache content with the new output vector for that round, and $\delta_i = 0$,  the after the first round, the node does not cache anything after any round.

We now describe the usage of Lemma \ref{lem:func_updates_var2}. } Toward this, we note that ${\bf y}_i^{(r)}$ can be written as 
\begin{eqnarray}
{\bf y}_i^{(r)} & = & A_i{\bf m}^{(r)}, 1 \leq i \leq N,
\end{eqnarray}
where the matrix $A_i$ is uniquely determined given the LNC $\mathcal{C}$. We assume that matrix $A_i$ has the form 
\begin{eqnarray}
A_i & = & \left[ \begin{array}{c} A_{i, f_1} \\ A_{i, f_2} \\ \vdots \\ A_{i, f_{\beta_i}}
\end{array} \right], \{f_1, f_2, \ldots, f_{\beta_i}\}  = \mathcal{N}_{out}(i),
\end{eqnarray}
where $A_{i, f_{\ell}}$ is such that ${\bf y}_{i, f_{\ell}}^{(r)} = A_{i, f_{\ell}}{\bf m}^{(r)}, f_{\ell} \in \mathcal{N}_{out}(i)$.   Consider the indices $\mathcal{N}_{in}(i) = \{e_1, e_2, \ldots, e_{\alpha_i}\}$ of in-neighbors of $V_i$. 
If $\delta_i = 0$, in-neighbor $V_{e_{\ell}}, e_{\ell} \in \mathcal{N}_{in}(i)$ simply sends, like in round $1$, ${\bf y}_{e_{\ell},i}^{(r)}$   to $V_i$. However, if $\delta_i = 1$, $V_{e_{\ell}}, e_{\ell} \in \mathcal{N}_{in}(i)$ encodes ${\bf y}_{e_{\ell},i}^{(r)}$ using scheme of Lemma \ref{lem:func_updates_var2} assuming destination side-information ${\bf y}_{i}^{(r-1)}$.
The function-updates coding problem that arises here is shown in Fig. \ref{fig:coding_updates_opt}.  It is straight forward to see that the problem in Fig. \ref{fig:coding_updates_opt} is an instance of the one in Fig. \ref{fig:coding_updates_var2}, and thus the scheme guaranteed by Lemma \ref{lem:func_updates_var2} can be applied to possibly reduce communication cost needed to compute  ${\bf y}_{i}^{(r)}$ at the node $V_i$. This completes the description of the scheme. An overview of the various steps at node $V_i$ that uses its cache is pictorially shown in Fig. \ref{fig:overview}.

\begin{figure}[h]
	\centering
	\includegraphics[width=3.5in]{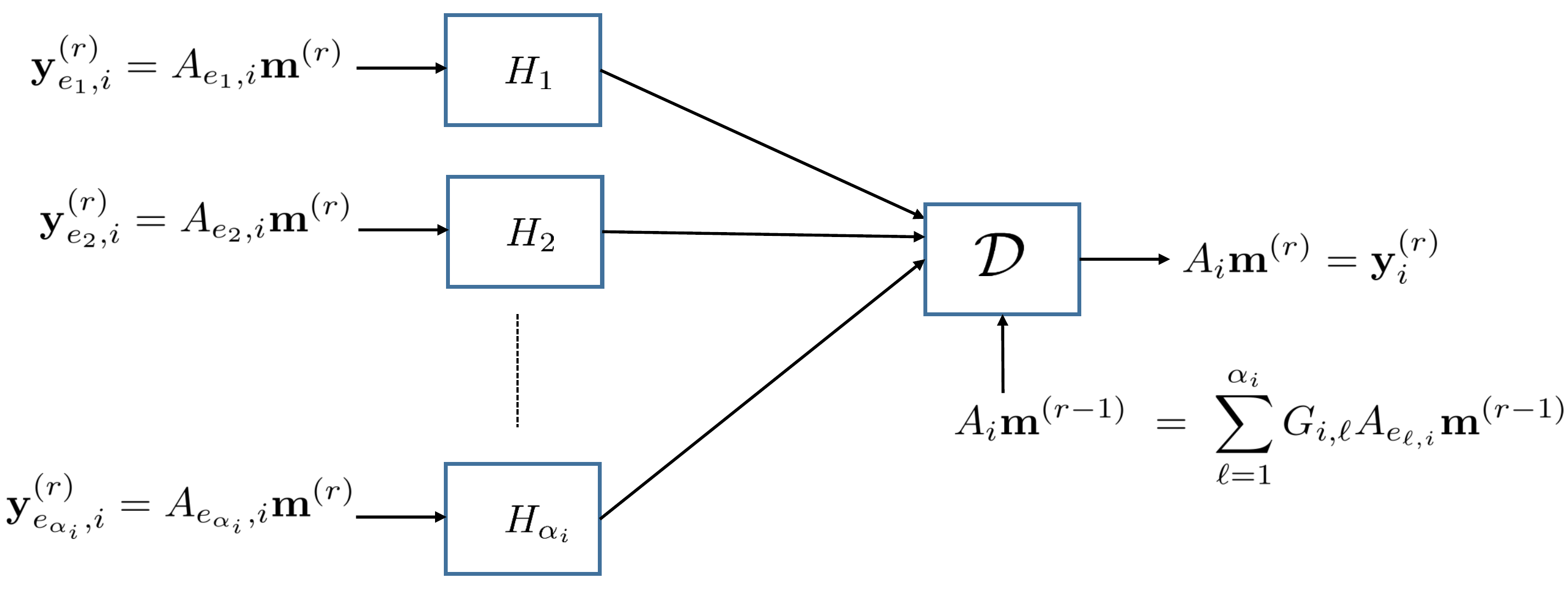}
	\caption{The function-updates coding problem at the in-neighbors of $V_i$ during round $r$. $V_i$ has cached its output from round $(r-1)$, and this is available as side-information during decoding its output for round $r$. The problem is an instance of the one shown in Fig. \ref{fig:coding_updates_var2}.}
	\label{fig:coding_updates_opt}
\end{figure}

\begin{figure*}[h]
	\centering
	\includegraphics[width=5in]{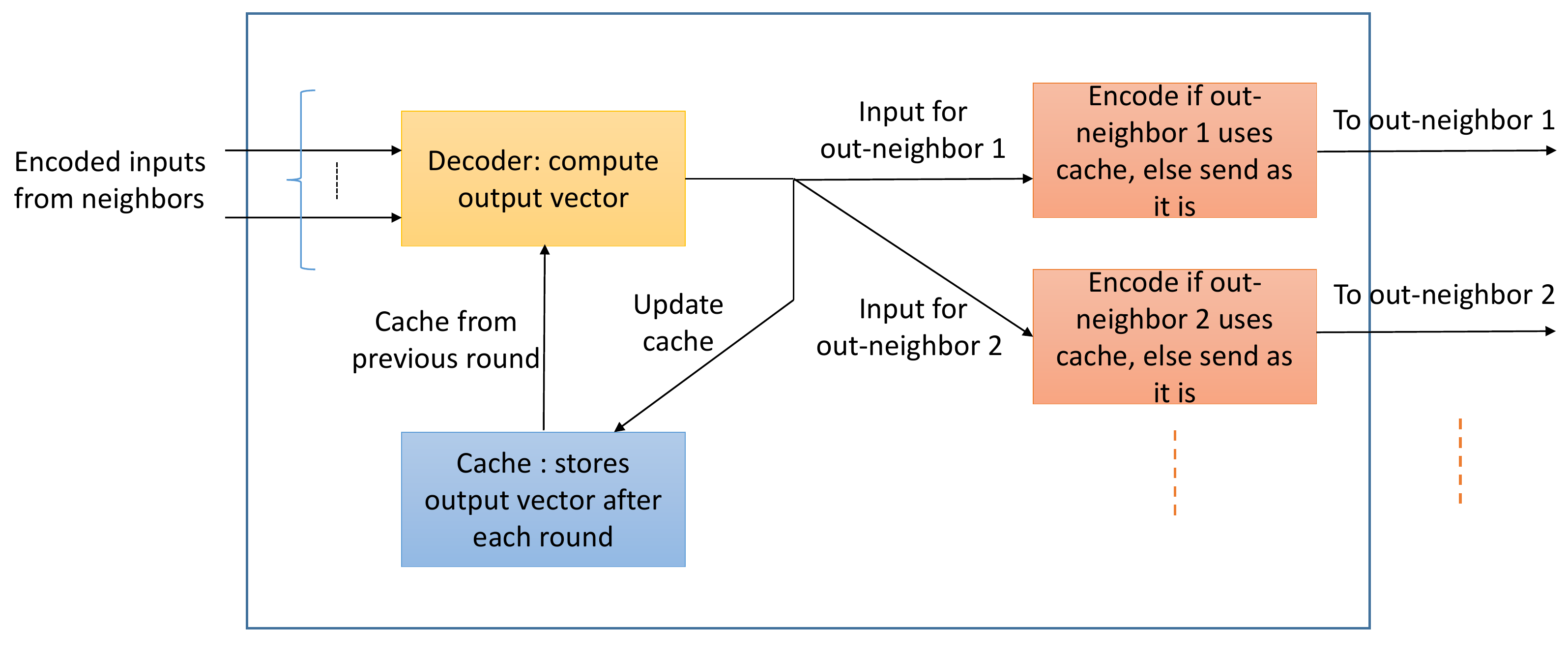}
	\caption{An overview of the steps performed by a non-source node that uses its cache storing its vector after each round.}
	\label{fig:overview}
\end{figure*}

\vspace{0.1in}

\begin{note}
	In our discussion, we  ignore the issue of how the encoder and decoders are designed for the various nodes in the network. A simple approach is to assume that every node has global knowledge of the LNC $\mathcal{C}$, and in addition, each node also knows whether any of its out-neighbors uses its cache or not. Both these assumptions ensure that each node can independently design the encoder and decoders necessary for using the function updates scheme. We note that the assumption is not a  harsh one, since the same LNC is reused for $M$ rounds. In practice, it might take a few initial rounds of network gossip (via different paths that are not available for data transmission) for establishing global knowledge of $\mathcal{C}$, and  we can make use of the function updates scheme after these initial rounds.  Note that in the current discussion, we assume that the function updates scheme is usable right from round $2$.
\end{note}

\section{Optimizing the Storage and Communication Cost for the Scheme} \label{sec:opt}

We now discuss the storage and communication cost associated with the scheme in Section \ref{sec:scheme}. We compute the costs separately for the first round, and for any round beyond the first. We use $\rho_S(r)$ and $\rho_C(r)$ to respectively denote the storage and communication cost of round $r, 1 \leq r \leq M$. We are interested in optimizing the overall cost for round $r$, given by $\rho(r) = \rho_S(r) + \rho_C(r)$.

\vspace{0.1in}

\emph{Cost of Round $1$:} For any edge $(i, j) \in \mathcal{E}$, we define $\sigma_{i, j}$ as the number of bits sent by node $V_i$ to node $V_j$, i.e., 
\begin{eqnarray} \label{eq:length}
\sigma_{i, j} & = & \text{length}\left({\bf y}_{i, j}^{(1)}\right)\log_2(q).
\end{eqnarray}
There is no storage cost in round $1$ since caches are initially empty. The overall cost for round $1$, which is the total communication cost, is given by
\begin{eqnarray}
\rho(1) & = & \rho_C(1) \ = \sum_{(i, j) \in \mathcal{E}} f_{i, j}\left(\sigma_{i, j}\right),
\end{eqnarray}
where recall that $f_{i,j}(s), s\geq 1$ denotes the cost to transmit, in a single round, a packet of size $s$ bits along the directed edge $V_i \rightarrow V_j$.

\vspace{0.1in}

\emph{Cost of Round $r, 2 \leq r \leq M$:} We first calculate cost incurred to calculate the output at a single node $V_i , 2 \leq i \leq N$. If $\delta_i = 0$ (i.e., $V_i$ does not use cache), cost of calculating the output of $V_i$ is simply the cost associated with the incoming edges of $V_i$. This is given by $\sum_{j \in \mathcal{N}_{in}(i)}f_{ j,i}\left(\sigma_{j,i}\right)$, where $\sigma_{j, i}$ is  defined\footnote{Note that length of the output vector of any node remains the same for any round. This is because, we effectively use the same LNC $\mathcal{C}$ for all $M$ rounds. The function update scheme should be thought of as an alternative method, when compared to the first round, to enable the usage of the code $\mathcal{C}$} as in \eqref{eq:length}. If $\delta_i = 1$ (i.e., $V_i$ uses cache), cost of calculating the output of $V_i$ includes cache cost as well as communication cost associated with the incoming edges of $V_i$ for sending encoded updates. The cache cost of $V_i$ for round $r$ is upper bounded by $f_i(\sigma_{i})$, where 
\begin{eqnarray} \label{eq:length2}
\sigma_{i} & = &\text{length}\left({\bf y}_{i}^{(r-1)}\right)\log_2(q).
\end{eqnarray}
Note that $f_i(\sigma_{i})$ is only an upper bound on the cache cost, since the entries of {${\bf y}_{i}^{(r-1)}$} may be linearly dependent\footnote{We do not bother compressing { ${\bf y}_{i}^{(r-1)}$}, by eliminating linear dependence, since it complicates the objective function for our cost optimization problem.} (for example the same linear combination may be sent to two output links). Also, note that $\sigma_{i}$ is the same for any $r, 1 \leq r \leq M$. Furthermore, $\sigma_{i}$ in \eqref{eq:length2} is also given by 
\begin{eqnarray} \label{eq:length3}
\sigma_{i} & = & \sum_{\ell \in \mathcal{N}_{out}(i)}\sigma_{i, \ell}. 
\end{eqnarray}
The communication cost, using Lemma \ref{lem:func_updates_var2}, is upper bounded by { $\sum_{j \in \mathcal{N}_{in}(i)}f_{j, i}\left(2\epsilon\right)$}, where $\epsilon = \text{Hamming wt.}({\bf m}^{(r-1)} - {\bf m}^{(r-1)}), 2 \leq r \leq M$. Note that the communication cost in this case is only an upper bound, since we ignore the min-function calculation in Lemma \ref{lem:func_updates_var2}. Thus the overall cost of round $r, 2 \leq r \leq M$ is upper bounded as follows:
{
\begin{eqnarray}
\rho(r) & \leq & \rho*(r) \triangleq \sum_{i = 2}^{N}\left( (1-\delta_i)\sum_{j \in \mathcal{N}_{in}(i)}f_{ j, i}\left(\sigma_{i, j}\right) + \nonumber \right.\\ 
& & \left.\delta_i\left(f_i(\sigma_{i}) +  \sum_{j \in \mathcal{N}_{in}(i)}f_{j, i}\left(2\epsilon\right) \right)\right). \label{eq:long1}
\end{eqnarray}
The overall cost $\Psi_{\mathcal{S}}$ of the scheme $\mathcal{S}$ for all $M$ rounds is given by 
\begin{eqnarray}
\Psi_{\mathcal{S}} & = & \sum_{r = 1}^{M} \rho(r) \\
& \leq & \sum_{(i, j) \in \mathcal{E}} f_{i, j}\left(\sigma_{i, j}\right) + (M-1)\rho*(r)  \\
& \triangleq & \Psi^*(\Sigma, \Delta),
\end{eqnarray}
}
where {$\rho*(r)$ is as defined in \eqref{eq:long1}} and $\Sigma$ is the multi-set $\{\sigma_{j, i}, 2 \leq i \leq N, j \in \mathcal{N}_{in}(i)\}$, and $\Delta$ is the multi-set $\{\delta_i, 2 \leq i \leq N\}$.

\subsection{The Optimization Problem of Interest}

Recall that the scheme $\mathcal{S}$ is designed for a given directed acylic graph$(\mathcal{N}, \mathcal{E})$ and for fixed values of $B, \epsilon, M$. The scheme involves $1)$ picking an LNC $\mathcal{C}$ such that $\eqref{eq:decode_rd1}$ is satisfied, and $2)$ deciding the values of the elements of $\Delta = \{\delta_i, 2 \leq i \leq N\}$. Note that the LNC $\mathcal{C}$ automatically fixes the values of the elements of  $\Sigma = \{\sigma_{j, i}, 2 \leq i \leq N, j \in \mathcal{N}_{in}(i)\}$. The cost of the scheme is upper bounded by $\Psi_{\mathcal{S}} \leq \Psi^*(\Sigma, \Delta)$.

\vspace{0.1in}

\begin{thm}
Let $\widehat{\Sigma}$ and $\widehat{\Delta}$ denote the multi-sets $\{\widehat{\sigma}_{j, i} \in \mathbb{R}, 2 \leq i \leq N, j \in \mathcal{N}_{in}(i)\}$ and $\{\widehat{\delta}_i \in \{0, 1\}, 2 \leq i \leq N\}$. Consider the following optimization problem for given values of $\epsilon, M, B$ and a given DAG $(\mathcal{N}, \mathcal{E})$:
\begin{eqnarray}\label{eq:opt_orig}
& &\text{minimize } \Psi^*(\widehat{\Sigma}, \widehat{\Delta}) \\ 
& &\text{  subject to } \\
& & \widehat{\delta}_{\ell} \in \{0, 1\}, 2 \leq \ell \leq N \label{eq:integral_const} \\
& & \widehat{\sigma}_{j, i} \geq \widehat{\mu}_{j, i}^{(t)} \geq 0 \label{eq:max_const} \\
& & \sum_{j \in \mathcal{N}_{out}(i)} \widehat{\mu}_{i, j}^{(t)} - \sum_{j \in  \mathcal{N}_{in}(i)} \widehat{\mu}_{j, i}^{(t)} = \theta_i^{(t)}, \nonumber\\
& & 1 \leq i \leq N, t \in \{N - L + 1, \ldots, N\} \label{eq:flow_conservation}
\end{eqnarray}
where
\begin{eqnarray}
\theta_i^{(t)} & = & \left\{ \begin{array}{c} B, i = 1 \\ -B, i = t \\ 0, \text{else} \end{array} \right. .
\end{eqnarray}
Then corresponding to any feasible solution $\widehat{\Sigma}, \widehat{\Delta}$, there exists an achievable scheme $\mathcal{S}$ (for the given system parameters) operating over $\mathbb{F}_q, q \geq 2\epsilon B^{2\epsilon}$ such that 
\begin{enumerate}
	\item the number of symbols per round generated by node $V_i$ for node $V_j$ is given by $\sigma_{i, j} = \left \lceil \widehat{\sigma_{i, j}}/\log (q) \right \rceil$
	\item the number of symbols per round sent by node $V_i$ to node $V_j$ is given  by $2 \delta_j \epsilon + ( 1- \delta_j) \sigma_{i, j}$, where $\delta_j  = \widehat{\delta}_j$.
	\item the number of symbols cached per round by node $V_j$ is upper bounded by $\delta_j \sigma_j$, where  $\sigma_j$ is given by \eqref{eq:length3}, and 
	\item the cost of the scheme is upper bounded by $\Psi_{\mathcal{S}} \leq \Psi^*(\Sigma, \Delta)$.
\end{enumerate}
\end{thm}
\begin{proof}
Follows by combining the proof of Theorem $1$ of \cite{Lun} with the properties of our scheme in Section \ref{sec:scheme}.
\end{proof}	

  It is worth noting that in the optimization problem (\ref{eq:opt_orig}), condition \eqref{eq:max_const} ensures that the rate of per-terminal virtual flows, denoted by $\widehat{\mu}_{j, i}^{(t)}$, for all $t \in \{N - L + 1, \ldots, N\}$, are no more than $\widehat{\sigma}_{j, i}$, the rate of coded flow which is transporting packets to all terminals in session. In addition, condition \eqref{eq:flow_conservation} enforces that a solution to the optimization problem satisfies the per-terminal flow conservation at each intermediary node. That is, no amount of flow is moving in or out of intermediary nodes.\color{black}
\subsection{Optimization Relaxation}
The optimization problem given in \eqref{eq:opt_orig} is a mixed-integer problem, and hence, solving for it directly is NP-Hard. To construct a convex relaxation of the \eqref{eq:opt_orig} problem, we use a technique known as multi-linear relaxation \cite{calinescu2011}, \cite{sigmetrics}, in which, we suppose that variables $\delta_i$ are independent Bernoulli random variables, with joint probability distribution $\nu$ defined in $\{0,1\}^{N}$. Let
$\kappa_i$ be the marginal probability that node $i$ chooses to cache. We have
$$\kappa_i = \mathbb{E}_{\nu}[\delta_i].$$

We further note that since $\Psi$ is linear in $\Delta$, it can be verified that  
$$\mathbb{E}_{\nu}[\Psi(\Sigma,\Delta)] = \Psi(\Sigma,\mathbb{E}_{\nu}[\Delta]) = \Psi(\Sigma,K),$$
	where $K = [\kappa_i]_{i \in \mathcal{N}} \in [0,1]^{N}$. 
	
Using multi-linear relaxation, the mixed-integer problem is turned into a continuous one, in which we seek to minimize $\mathbb{E}_{\nu}[\Psi(\Sigma,\Delta)] $ subject to the feasibility constraints.

 In addition to this relaxation, we further relax the max constraint given in \eqref{eq:max_const} using $\ell^n$-norm approximation. That is, we use $\widetilde{\sigma}_{j,i} = \left(\sum\limits_{t\in T}\left({\mu}^{(t)}_{j,i}\right)^n\right)^{\frac{1}{n}}$ as an approximation of $\widehat{\sigma}_{j,i}$. We note that a code with rate $\widetilde{\sigma}_{j,i}$ on each link $(j,i)$ exists for any feasible solution of $\widehat{\sigma}_{j,i}$, since we have $\widetilde{\sigma}_{j,i}\geq \widehat{\sigma}_{j,i}$ for all $n>0$. Also, $\widetilde{\sigma}_{j,i}$ approaches $\widehat{\sigma}_{j,i}$, as $n\rightarrow \infty$. Hence, we assume that $n$ is large in (\ref{eq:opt_orig}). 

Thus, the relaxed optimization problem follows
\begin{eqnarray}\label{eq:opt_relaxed}
& &\text{minimize } \Psi(\Sigma, K)\\ 
& &\text{  subject to } \\
& & \kappa_{\ell} \in [0, 1], 2 \leq \ell \leq N \label{eq:const_delta} \\
& & \sum_{j \in \mathcal{N}_{out}(i)} \mu_{i, j}^{(t)} - \sum_{j \in  \mathcal{N}_{in}(i)} \mu_{j, i}^{(t)} = \theta_i^{(t)}, \nonumber\\& & 1 \leq i \leq N, t \in \{N - L + 1, \ldots, N\}\label{eq:const_conservation}\\
& & \mu_{i, j}^{(t)}\geq 0 \label{eq:const_mu}
\end{eqnarray}
where 
\begin{eqnarray*}
	&&\Psi({\Sigma}, {K}):= \sum_{(i, j) \in \mathcal{E}} f_{i, j}\left(\left(\sum\limits_{t\in T}\left({\mu}^{(t)}_{i,j}\right)^n\right)^{\frac{1}{n}}\right) + \\ &&(M-1)\sum_{i = 2}^{N}\biggl( (1-{\kappa}_i)\sum_{j \in \mathcal{N}_{in}(i)}f_{ j, i}\left(\left(\sum\limits_{t\in T}\left({\mu}^{(t)}_{j,i}\right)^n\right)^{\frac{1}{n}}\right) \\&+&{\kappa}_i\biggl(f_i(\widetilde{\sigma}_{i}) + \sum_{j \in \mathcal{N}_{in}(i)}f_{j, i}\left(2\epsilon\right)\biggr)\biggr),
\end{eqnarray*}
and, $\widetilde{\sigma}_{i} =\sum\limits_{\ell \in \mathcal{N}_{out}(i)}\left(\sum\limits_{t\in T}\left({\mu}^{(t)}_{i,\ell}\right)^n\right)^{\frac{1}{n}}$.

The relaxed problem is now a convex problem which can be solved using any standard optimization algorithm. In the next section, we describe primal-dual algorithms to solve this problem, which can be implemented in an offline or online setting. We later discuss how to recover solutions with integer $\kappa_i$ from the fractional solutions. 
\subsection{Primal-Dual Algorithms}
Following the scheme used in \cite{Lun}, we apply primal-dual scheme to obtain optimal solutions for the relaxed optimization problem. We note that $\Psi (\Sigma,K)$ is strictly convex in $\mu$'s, but linear in $\kappa$'s. 

The Lagrangian function of $\Psi (\Sigma,K)$ is established as
\begin{eqnarray}
& & L(\Sigma,K,P,\Lambda,\Gamma) = \Psi (\Sigma,K)+\sum\limits_{i\in \mathcal{N},t\in T}p_i^{(t)}(y_i^{(t)} -\theta_i^{(t)}) \nonumber\\&& -\sum\limits_{(i,j)\in \mathcal{E},t\in T}\lambda_{i,j}^{(t)}\mu_{i,j}^{(t)} - \sum\limits_{i=2}^{N} \kappa_i\gamma_i^- + \sum\limits_{i=2}^{N} (\kappa_i-1)\gamma_i^+.
\end{eqnarray}
where $\gamma_i^-$, $\gamma_i^+$, $p_i^{(t)}$ and $\lambda_{i,j}^{(t)}$ are, respectively, the Lagrange multipliers associated with constraints \eqref{eq:const_delta},\eqref{eq:const_delta}, \eqref{eq:const_conservation}, and \eqref{eq:const_mu}. In addition,
\begin{equation}
y^{(t)}_{i}:= \sum\limits_{\{j\in \mathcal{N}_{out}(i)\}}  \mu^{(t)}_{i,j} - \sum\limits_{\{j\in \mathcal{N}_{in}(i)\}}  \mu^{(t)}_{j,i} . \label{eq:yi}
\end{equation}

Let $(\widehat{\Sigma},\widehat{K},\widehat{P},\widehat{\Lambda},\widehat{\Gamma})$ be a solution for the relaxed problem, then the following Karush-Kuhn-Tucker conditions can be verified to hold:
\begin{equation}
\begin{split}
\frac{\partial L(\widehat{\Sigma},\widehat{K},\widehat{P},\widehat{\Lambda},\widehat{\Gamma})}{\partial \mu^{(t)}_{i,j}} = \frac{\partial \Psi(\widehat{\Sigma},\widehat{K})}{\partial \mu^{(t)}_{i,j}} + (\widehat{p}^{(t)}_{i}-\widehat{p}^{(t)}_{j})- \widehat{\lambda}^{(t)}_{ij} = 0,\\
\forall (i,j)\in \mathcal{E},t\in T,
\end{split}
\label{eq:kkt-lag_mu}
\end{equation}
\begin{equation}
\frac{\partial L(\widehat{\Sigma},\widehat{K},\widehat{P},\widehat{\Lambda},\widehat{\Gamma})}{\partial \kappa_{i}} = \frac{\partial \Psi(\widehat{\Sigma},\widehat{K})}{\partial\kappa_{i}} + \gamma_i^+ - \gamma_i^-  = 0,\quad
\forall i \in \mathcal{N},
\label{eq:kkt-lag_delta}
\end{equation}
\begin{equation}
\begin{split}
 \sum_{j \in \mathcal{N}_{out}(i)} \widehat{\mu}_{i, j}^{(t)} - \sum_{j \in  \mathcal{N}_{in}(i)} \widehat{\mu}_{j, i}^{(t)} = \theta_i^{(t)}, \quad\forall i \in \mathcal{N}, t \in T
\end{split}
\label{eq:kkt-conserv}
\end{equation}
\begin{equation}
 \widehat{\mu}^{(t)}_{i,j}\geq 0,\qquad \widehat{\lambda}^{(t)}_{i,j}\geq 0, \quad\forall (i,j)\in \mathcal{E},,t\in T,\label{eq:kkt-muandl}
\end{equation}
\begin{equation}
 0\leq \widehat{\kappa}_i \leq 1, \qquad\widehat{\gamma}_i^-\geq 0, \qquad \widehat{\gamma}_i^+\geq 0, \quad\forall i \in \mathcal{N} \label{eq:kkt_deltaandgamma}
\end{equation}
\begin{equation}
 \widehat{\mu}^{(t)}_{ij}\widehat{\lambda}^{(t)}_{ij}= 0,\quad \forall (i,j)\in \mathcal{E},t\in T,\label{eq:kkt-mul}
\end{equation}
\begin{equation}
 \widehat{\kappa}_i\widehat{\gamma}_i^- =0,\qquad (\widehat{\kappa}_i-1)\widehat{\gamma}_i^+ =0, \quad\forall i \in \mathcal{N}. \label{eq:kkt_deltagamma}
\end{equation}

The derivatives of $\Psi$ with respect to $\mu$'s and $\kappa$'s can be, respectively, obtained by
\begin{eqnarray}\label{eq:psi-deriv-mu}
	&&\frac{\partial \Psi(\Sigma,K)}{\partial \mu^{(t)}_{i,j}} =\left(\frac{\mu_{i,j}^{(t)}}{\widetilde{\sigma}_{i,j}}\right)^{n-1}\biggl(f'_{i,j}(\widetilde{\sigma}_{i,j})\nonumber\\&&+(\Sigma-1)\left[\kappa_if_i'(\widetilde{\sigma}_i)+(1-\kappa_j)f_{i,j}'(\widetilde{\sigma}_{i,j})\right]
	\biggr),
\end{eqnarray}
\begin{eqnarray}\label{eq:psi-deriv-kappa}
\frac{\partial \Psi(\Sigma,K)}{\partial \kappa_i}&=&(M-1)\biggl(f_i(\widetilde{\sigma}_i)\nonumber\\&+&\sum\limits_{j\in \mathcal{N}_{in}(i)}\left[f_{j,i}(2\epsilon)-f_{j,i}(\widetilde{\sigma}_{j,i})\right]
\biggr).
\end{eqnarray}

We can now specify the continuous-time primal-dual algorithm which is proved in Theorem. \ref{theorem:stability} to converge to the globally optimal solution of the relaxed problem, for any initial choice of $(\Sigma,K,P)$.

\begin{equation}
\dot{\mu}^{(t)}_{i,j} = k^{(t)}_{i,j}(\mu^{(t)}_{i,j})\left(-\frac{\partial \Psi}{\partial \mu^{(t)}_{i,j}} - q^{(t)}_{i,j}+ \lambda^{(t)}_{i,j}\right).
\label{eq:mudot}
\end{equation}
\begin{equation}
\dot{\kappa}_i = h_i(\kappa_i)\left(-\frac{\partial \Psi}{\partial \kappa_i} -\gamma_i^+ +\gamma_i^-\right).
\label{eq:deltadot}
\end{equation}
\begin{equation}
\dot{p}^{(t)}_{i} = g^{(t)}_{i}(p^{(t)}_{i})\left(y^{(t)}_{i}-\theta^{(t)}_{i}\right).
\label{eq:pdot}
\end{equation}
\begin{equation}
\dot{\lambda}^{(t)}_{i,j} = m^{(t)}_{i,j}(\lambda^{(t)}_{i,j})\left(-\mu^{(t)}_{i,j}\right)_{\lambda^{(t)}_{i,j}}^+.
\label{eq:lambdadot}
\end{equation}
\begin{equation}
\dot{\gamma}^-_{i} = \alpha_i({\gamma}^-_{i})\left(-\kappa_i\right)_{{\gamma}^-_{i}}^+.
\label{eq:gamma-dot}
\end{equation}
\begin{equation}
\dot{\gamma}^+_{i} = \beta_i({\gamma}^+_{i})\left(\kappa_i-1\right)_{{\gamma}^+_{i}}^+.
\label{eq:gamma+dot}
\end{equation}
where $k^{(t)}_{i,j}(\mu^{(t)}_{i,j})$, $h_i(\kappa_i)$, $g^{(t)}_{i}(p^{(t)}_{i})$, $m^{(t)}_{i,j}(\lambda^{(t)}_{i,j})$, $\alpha_i({\gamma}^-_{i})$, $\beta_i({\gamma}^+_{i})$ are positive, non-decreasing continuous functions of $\mu^{(t)}_{i,j}$, $\kappa_i$, $p^{(t)}_{i}$, $\lambda^{(t)}_{i,j}$, ${\gamma}^-_{i}$ and ${\gamma}^+_{i}$, respectively, and
\begin{equation}
q^{(t)}_{i,j} \triangleq p^{(t)}_{i} -p^{(t)}_{j}.\label{eq:qij}
\end{equation}
\begin{equation}
(y)_x^+ := \begin{cases}
\ y, \text{ if } x>0,\\
\ \max\{y,0\}, \text{ if } x\leq 0.
\end{cases}
\end{equation}

In the next theorem, it is shown that the primal-dual algorithm converges to a globally optimal solution of the relaxed problem, for any initial choice of $(\Sigma,K,P)$, provided that we initialize $\Lambda$, $\Gamma^-$ and $\Gamma^+$ with non-negative entries.

\begin{thm}
	\label{theorem:stability}
	The primal-dual algorithm described by eq.s (\ref{eq:mudot})-(\ref{eq:gamma+dot}) converges to a globally optimal solution of the relaxed problem, for any initial choice of $(\Sigma,K,P)$.
\end{thm}
\begin{proof}
See Appendix \ref{app:convergence}
\end{proof}
\subsection{Recovering Integer Solutions}
As shown in previous sections, the primal-dual algorithm converges to a globally optimal solution of the relaxed problem. In this section, we describe a randomized rounding scheme in which an integral solution for $\Delta$ is recovered, such that $\mathbb{E}_{\nu}[\Delta] = K$. Since the $\delta_i$'s are assumed to be independent Bernoulli random variables with parameter $\kappa_i$, each node $i$ can simply construct such an integer solution in a distributed manner by setting $\delta_i =1$ with probability $\kappa_i$, and $\delta_i =0$ otherwise. Hence, each node, independently from other nodes, can decide to cache or not based on the solution obtained by the primal-dual algorithm.

\subsection{Distributed Implementation of the Primal Dual Algorithm}
The continuous-time primal-dual algorithm given in \eqref{eq:mudot}-\eqref{eq:gamma+dot} can be easily discretized. We omit the discrete-time versions for brevity. The algorithm can be also implemented in a distributed manner as follows: as shown in \eqref{eq:psi-deriv-mu} and \eqref{eq:psi-deriv-kappa}, respectively, the derivatives of $\Psi$ with respect to $\mu$ and $\kappa$'s can be computed at each node using only local information. Hence, node $i$ can update variables $\mu_{i,j}^{(t)}$, $\kappa_i$, $p^{(t)}_i$, $\lambda_{i,j}^{(t)}$, $\gamma_i^-$, $\gamma_i^+$, $\forall j\in \mathcal{N}_{out}(i),t\in T$ at each iteration. After the convergence of the algorithm, then each node can use the rounding scheme, described in the previous section, to obtain $\delta_i$. Using random linear coding in a distributed manner 
\cite{Lun}, will result in a fully distributed approach.

\section{Simulation Experiments} \label{sec:exp}
In this section, we present our simulation results of the proposed Primal-Dual algorithm for various network topologies and scenarios. These results are based on the implementation of the described scheme in MATLAB environment.  We note that since, to the best of our knowledge, there are no competing schemes for the oblivious file update problem, we are not able to compare the performance of the proposed scheme with that of any competing schemes. Hence, we focus on the performance of our scheme under different network scenarios. In particular, we consider caching at the edge or at the peers, and compare the results in the case of 1) no caching (No Caching), 2) only at the edge (Edge Caching), 3) caching at the terminals (Peer Caching), 4) caching at both terminals and edge nodes (Edge+Peer Caching).

In addition, we consider different types of cost functions associated with caching functionality. We focus on caching with a) no cost, b)linear cost function, c) quadratic cost function. In all these cases, we consider the link costs to be of form 
\begin{equation}
f_{ij}(\sigma_{ij})= \frac{\sigma_{ij}}{c_{ij}-\sigma_{ij}} , \text{ for }\sigma_{ij} <c_{ij}, \forall (i,j)\in \mathcal{E},
\end{equation}
which gives the expected number of packets waiting for or under transmission at link $(i,j)\in\mathcal{E}$ under an
$M/M/1$ queuing model.  
\subsection{Simulation Details}
We consider three topologies for our experiments, as depicted in Fig. \ref{fig:topos}. We note that the capacity of each link in the illustrated topologies are shown next to the links. In the Butterfly Network (Fig. \ref{topo:butterfly}), the source node $s$ is sending packets to terminals $t_1$ and $t_2$. Nodes 3 and 4 have local storage, in addition to the terminals. In the Service Network (Fig. \ref{topo:sn}) the source node $s$ wishes to multicast coded packets to terminals $t_1$, $t_2$, and $t_3$. In this topology, node 2 as well as the terminals have local storage. Finally, in the CDN Network (Fig. \ref{topo:cdn}), the source node $s$ is sending packets to terminals $t_1$ and $t_2$. In this topology, nodes 5 and 6 have local storage, in addition to the terminals. 

In all the network scenarios, we assume the number of rounds $M=100$. Also, at the rounds $2\leq r\leq M$, the amount of change in the files is limited to 1\% of the file size, $B$.

\begin{figure*}[h]
	\begin{subfigure}[t]{0.33\textwidth}
		\centering
		\includegraphics[scale=0.35]{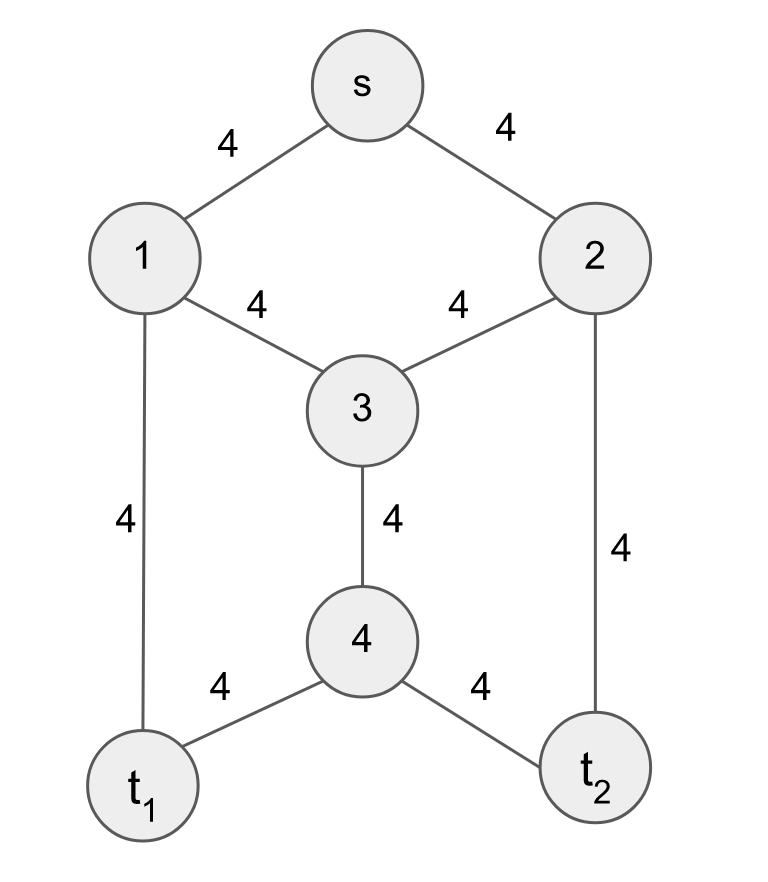}
		\caption{Butterfly Network.}
		\label{topo:butterfly}
	\end{subfigure}%
	\begin{subfigure}[t]{0.33\textwidth}
		\centering
		\includegraphics[scale=0.3]{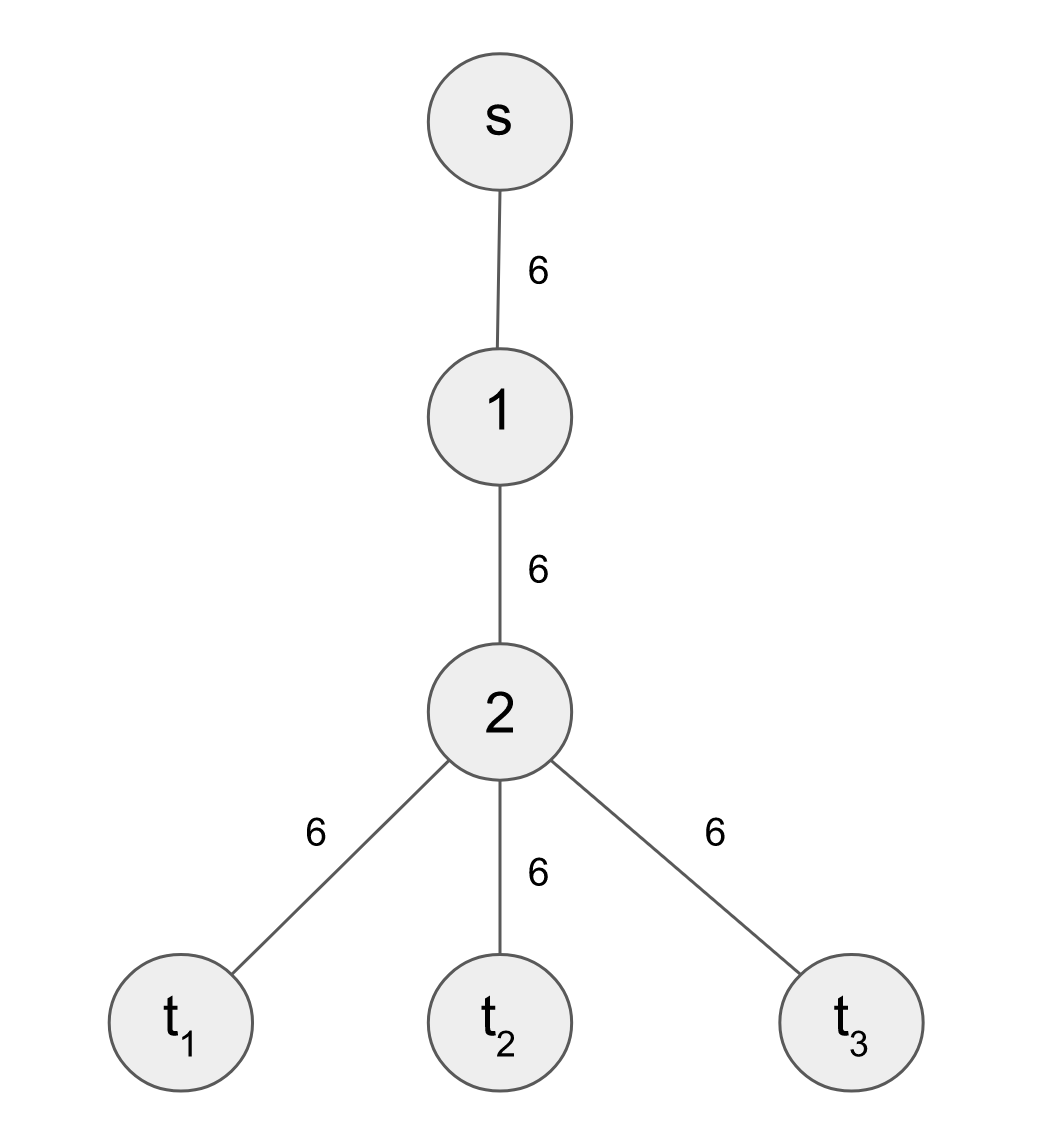}
		\caption{Service Network.}
		\label{topo:sn}
	\end{subfigure}
	\begin{subfigure}[t]{0.33\textwidth}
		\centering
		\includegraphics[scale=0.35]{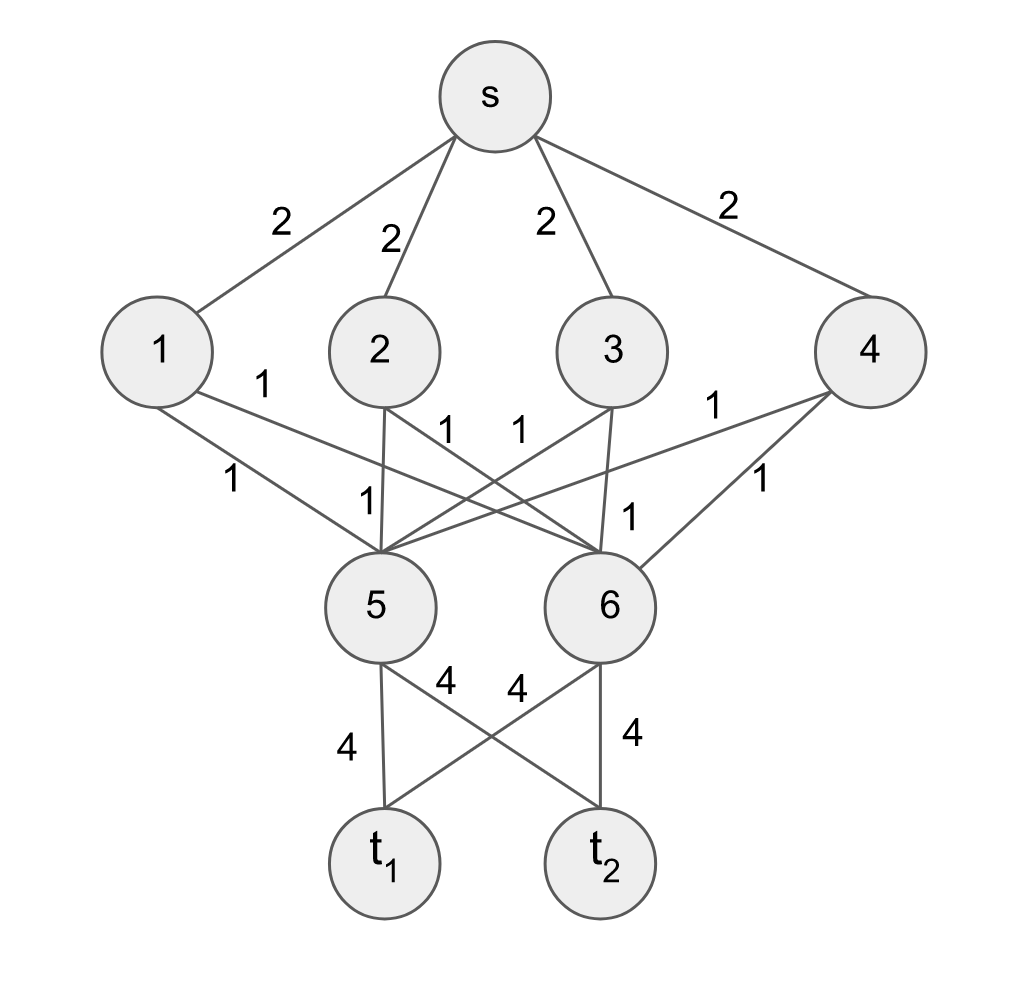}
		\caption{CDN Network.}
		\label{topo:cdn}
	\end{subfigure}	
	\caption{Experimented Topologies.}
	\label{fig:topos}
\end{figure*}

\subsection{Simulation Results}
Figure \ref{fig:butterfly} shows the convergence of the Primal-Dual Algorithm under different caching scenarios and different caching cost functions when $B=3.6$. Fig. \ref{fig:butterflyn} shows the convergence of the algorithm when there is no cost associated with caching functionality. In this case, the No Caching scenario has the highest total network cost, and the Peer+Edge Caching has the lowest. This is unsurprising, as caching will reduce the total network cost. Furthermore, the Peer Caching case has a lower cost comparing to the Edge Caching. 

Figures \ref{fig:butterflyl}, and \ref{fig:butterflyq} shows the convergence of the algorithm under linear and quadratic caching cost functions, respectively. In both cases, it can be verified that the caching at the edge is prevented due to the increased caching costs. That is why the performance of the No Caching case is similar to that of the Edge Caching one. Similarly, the performance of the Peer+Edge Caching case resembles that of the Peer Caching one. We note that as the $B$ increases and approaches the max-flow min-cut, the caching cost will be justified by the increasing link costs. Hence, it will be more likely to have caches in the edge.

This trend is similarly seen in Service Network and CDN Network, as seen in Figures \ref{fig:sn} and \ref{fig:cdn}, respectively. That is, edge caching is often prevented under low to medium traffic rate, when a (linear or quadratic) cost function associated with caching functionality is taken into account. However, as traffic is increased, the caching cost is justified to keep the link costs under control.

In Fig. \ref{fig:snn}, the performance of the proposed scheme on the Service Network under no caching cost function is illustrated. Similar to that on the Butterfly Network, the network cost under Peer Caching is lower than that under Edge Caching. Similarly, the network cost under the No Caching is the highest, while the it is the lowest under the Peer+Edge Caching. This is also true for the CDN Network, depicted in Fig. \ref{fig:cdnn}. However, it can been seen that the network cost under the Edge Caching is lower than that under Peer Caching. This behavior is the opposite of that for the Butterfly and Service Networks. 
\begin{figure*}[h]
	\begin{subfigure}[t]{0.33\textwidth}
		\centering
		\includegraphics[width =1\textwidth, height =.8\textwidth ]{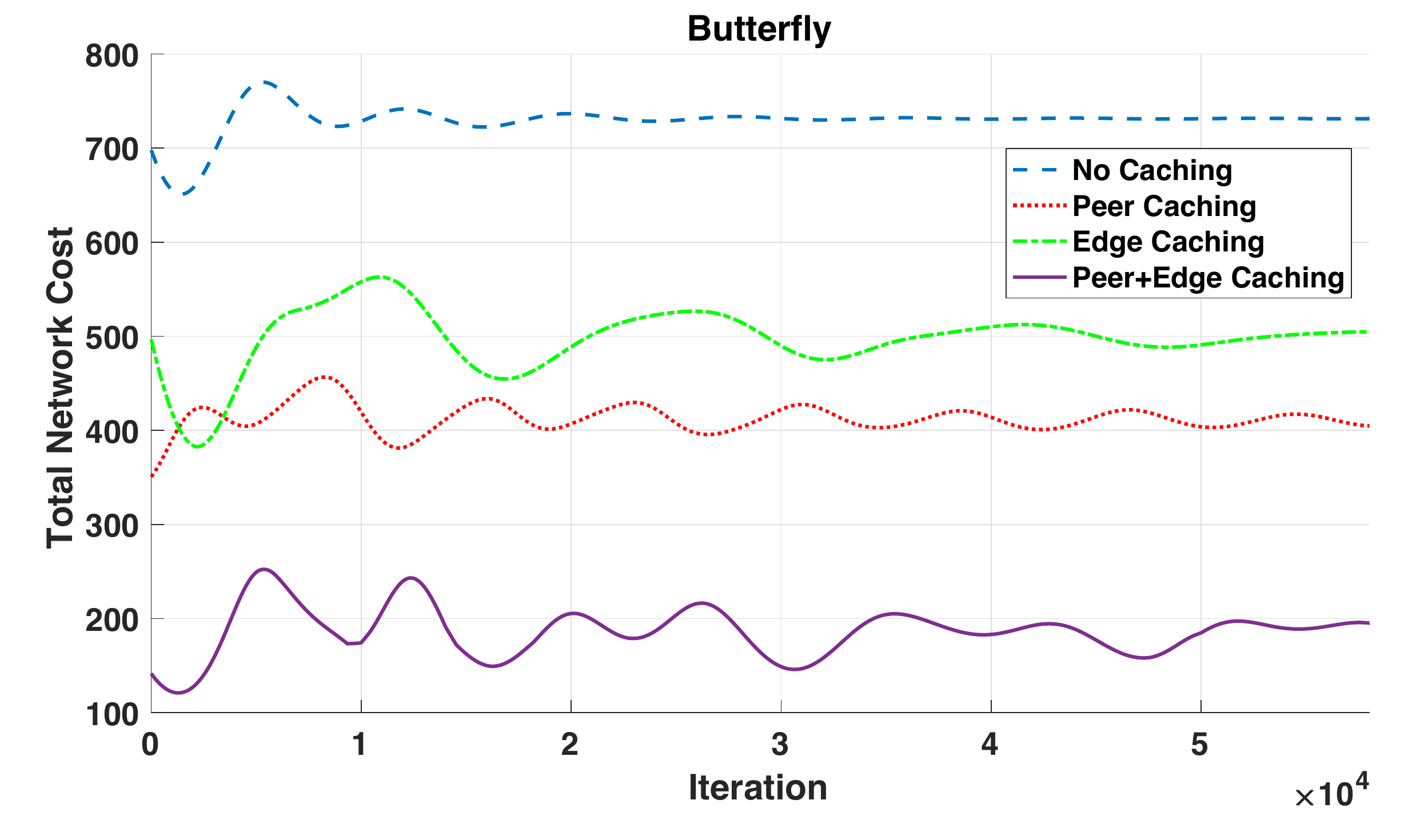}
		\caption{No Caching Cost.}
		\label{fig:butterflyn}
	\end{subfigure}%
	\begin{subfigure}[t]{0.33\textwidth}
		\centering
		\includegraphics[width =1\textwidth, height =.8\textwidth ]{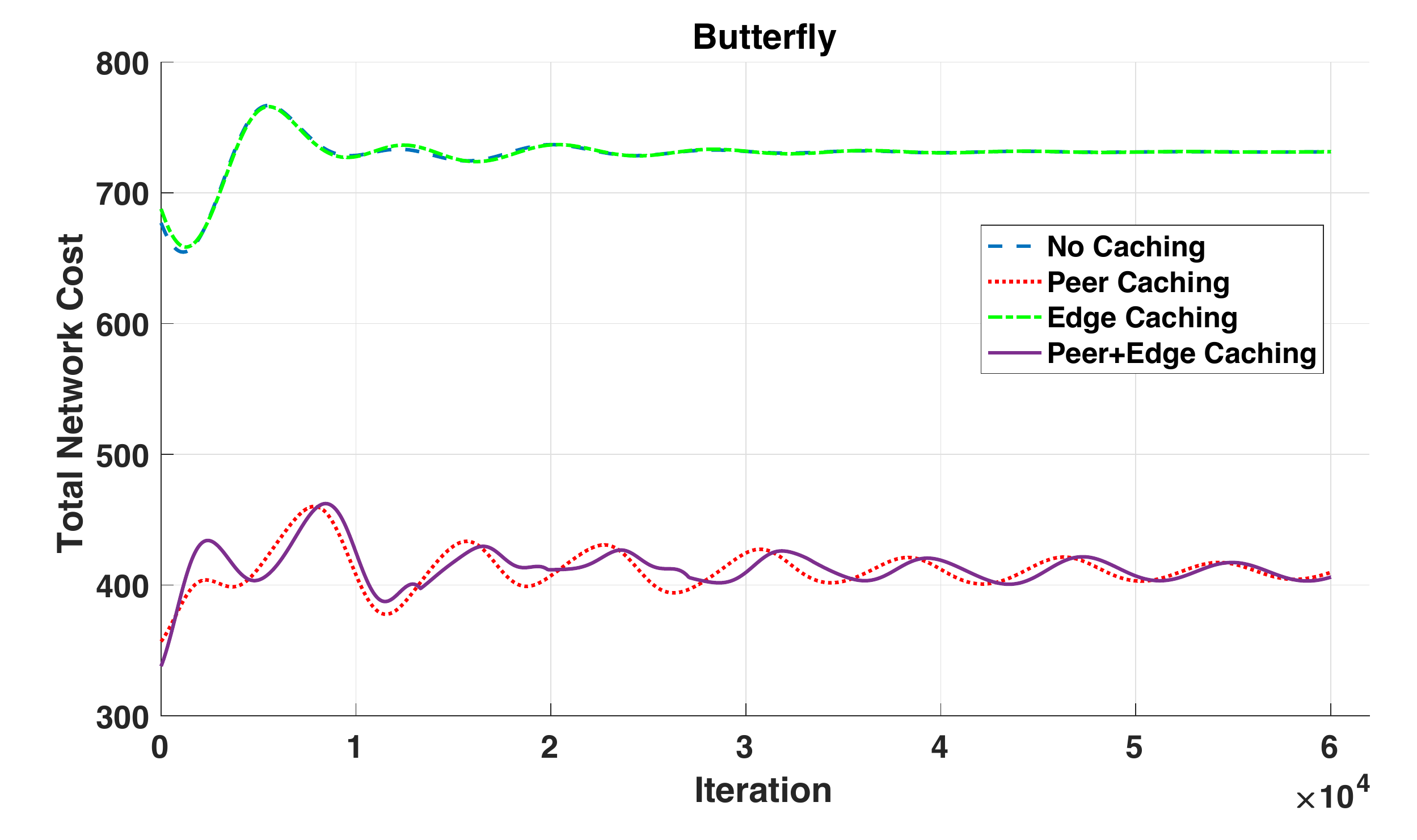}
		\caption{Linear Caching Cost.}
				\label{fig:butterflyl}
	\end{subfigure}
	\begin{subfigure}[t]{0.33\textwidth}
		\centering
		\includegraphics[width =1\textwidth, height =.8\textwidth ]{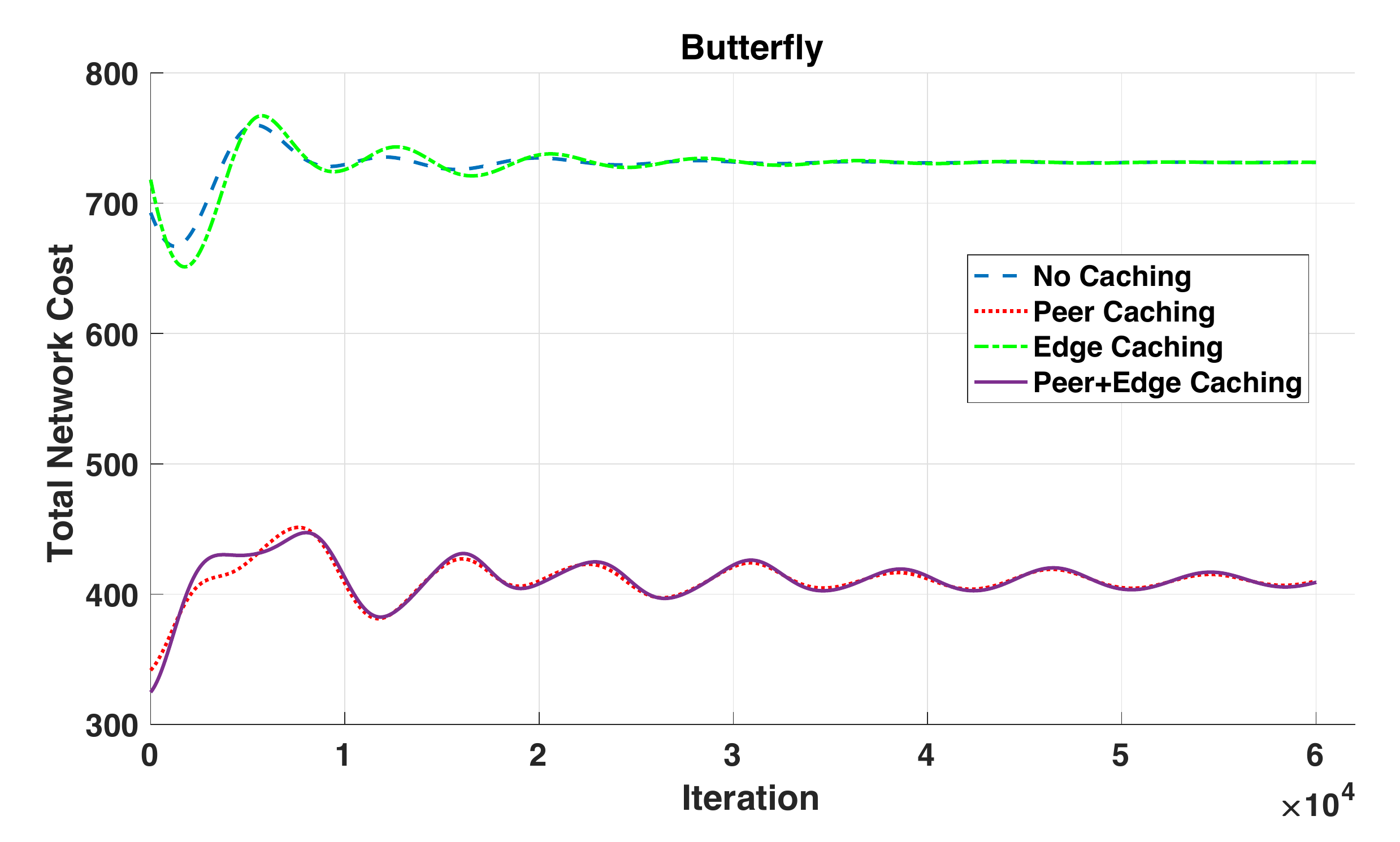}
		\caption{Quadratic Caching Cost.}
				\label{fig:butterflyq}
	\end{subfigure}	
	\caption{Convergence of the Primal-Dual Algorithm on the Butterfly Network under different caching cost functions and scenarios when $B=3.6$.}
	\label{fig:butterfly}
\end{figure*}
\begin{figure*}[h]
	\begin{subfigure}[t]{0.33\textwidth}
		\centering
		\includegraphics[width =1\textwidth, height =.8\textwidth ]{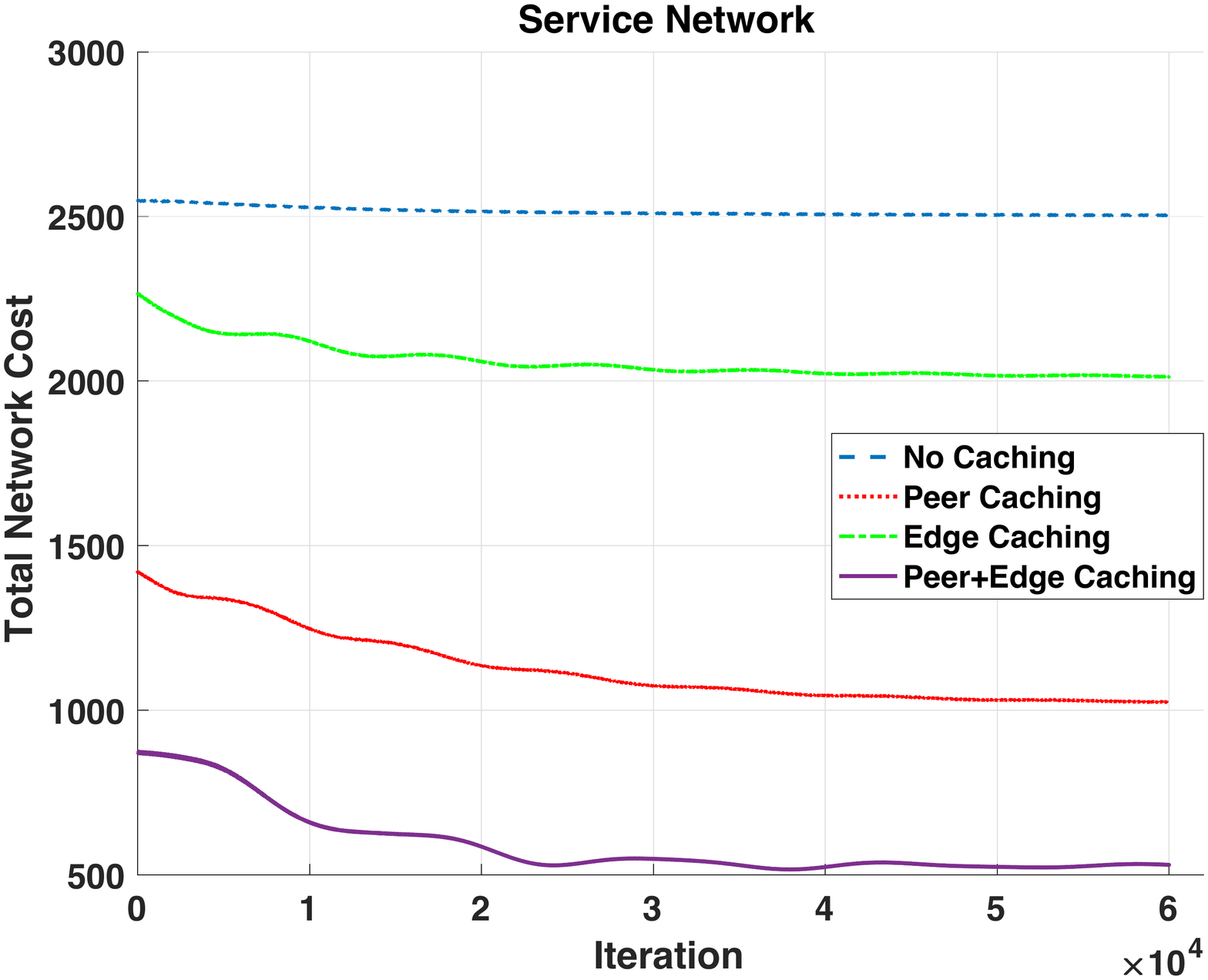}
		\caption{No Caching Cost.}
		\label{fig:snn}
	\end{subfigure}%
	\begin{subfigure}[t]{0.33\textwidth}
		\centering
		\includegraphics[width =1\textwidth, height =.8\textwidth ]{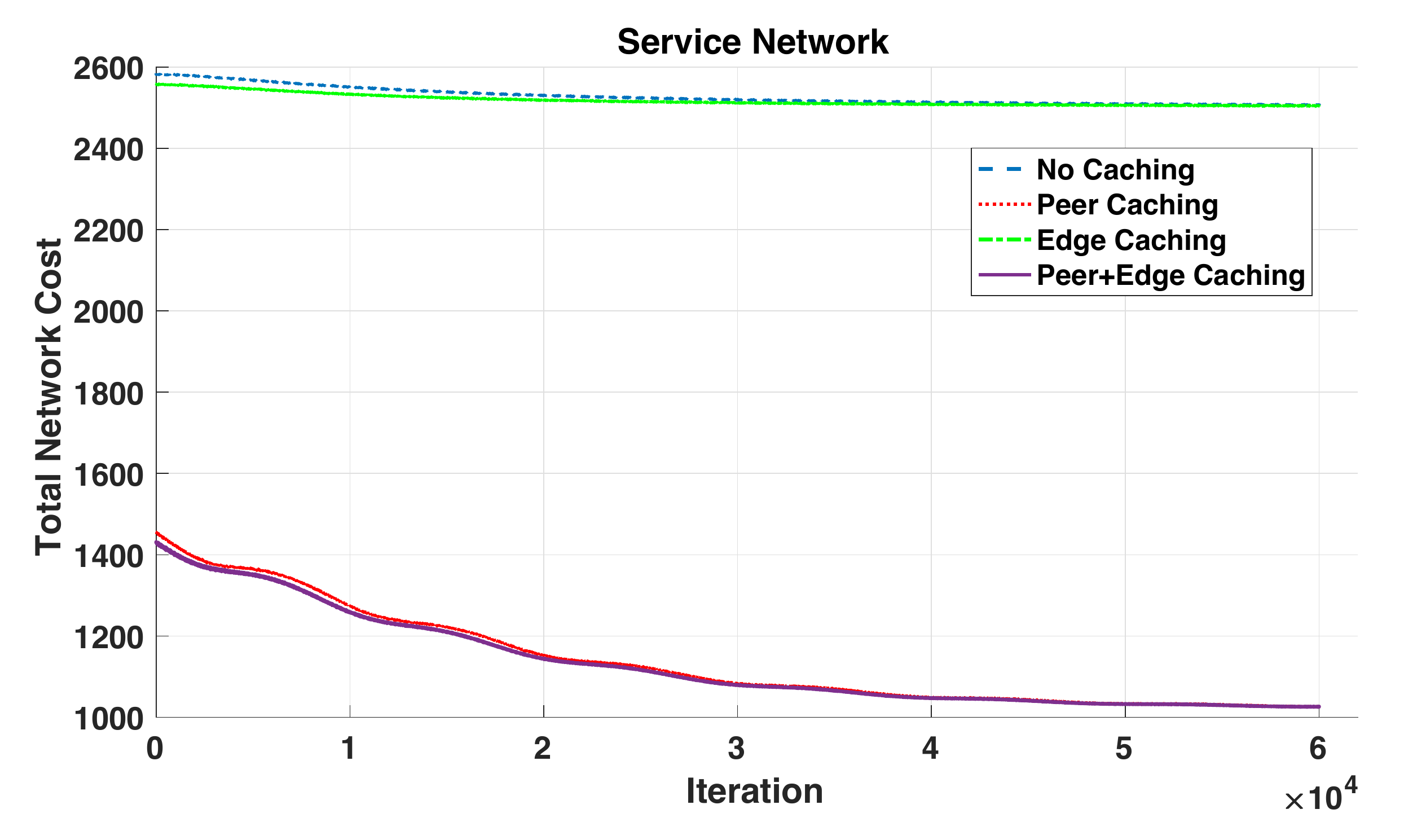}
		\caption{Linear Caching Cost.}
		\label{fig:snl}
	\end{subfigure}
	\begin{subfigure}[t]{0.33\textwidth}
		\centering
		\includegraphics[width =1\textwidth, height =.8\textwidth ]{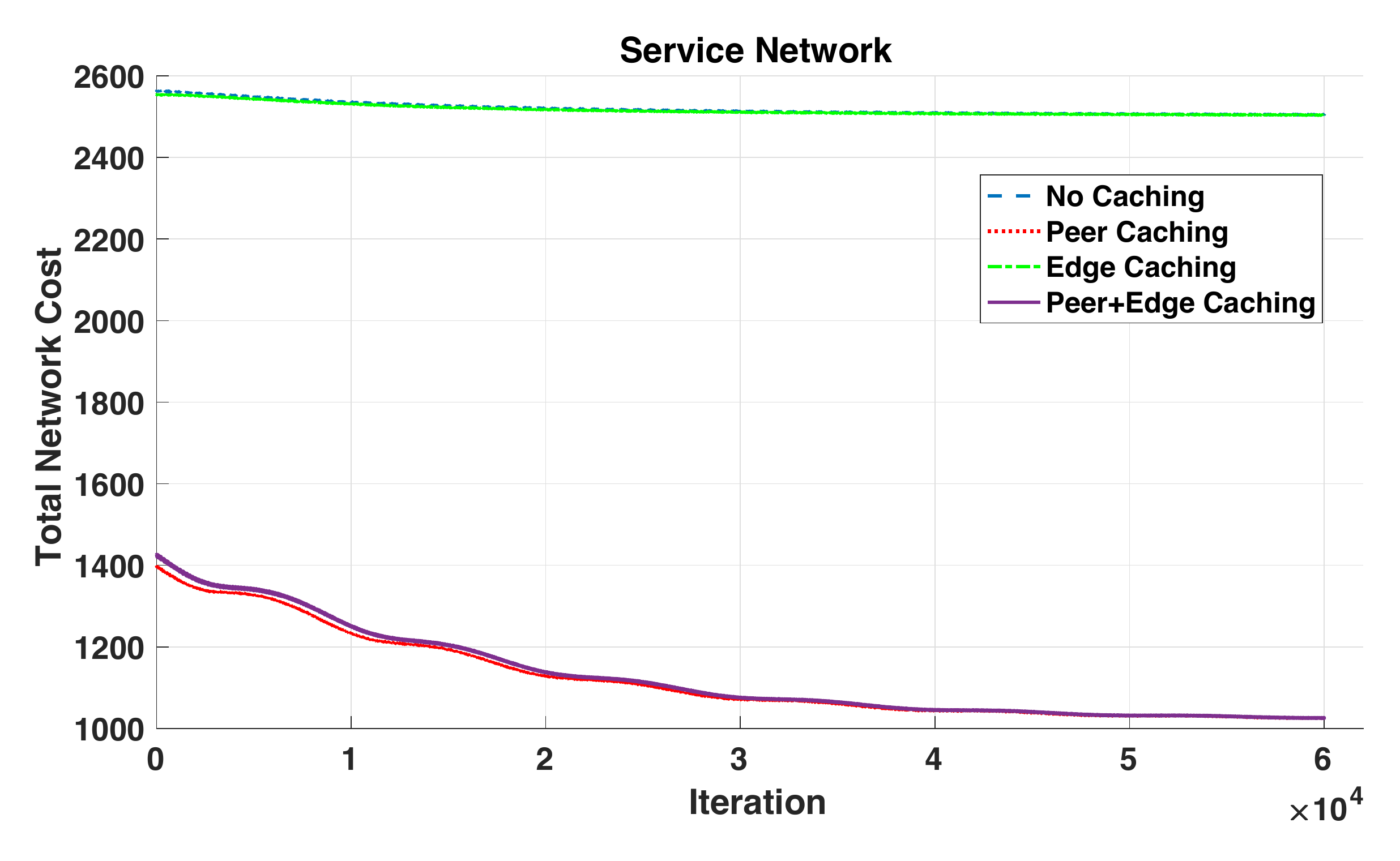}
		\caption{Quadratic Caching Cost.}
		\label{fig:snq}
	\end{subfigure}	
	\caption{Convergence of the Primal-Dual Algorithm on the Service Network under different caching cost functions and scenarios when $B=5$.}
	\label{fig:sn}
\end{figure*}
\begin{figure*}[h]
	\begin{subfigure}[t]{0.33\textwidth}
		\centering
		\includegraphics[width =1\textwidth, height =.8\textwidth ]{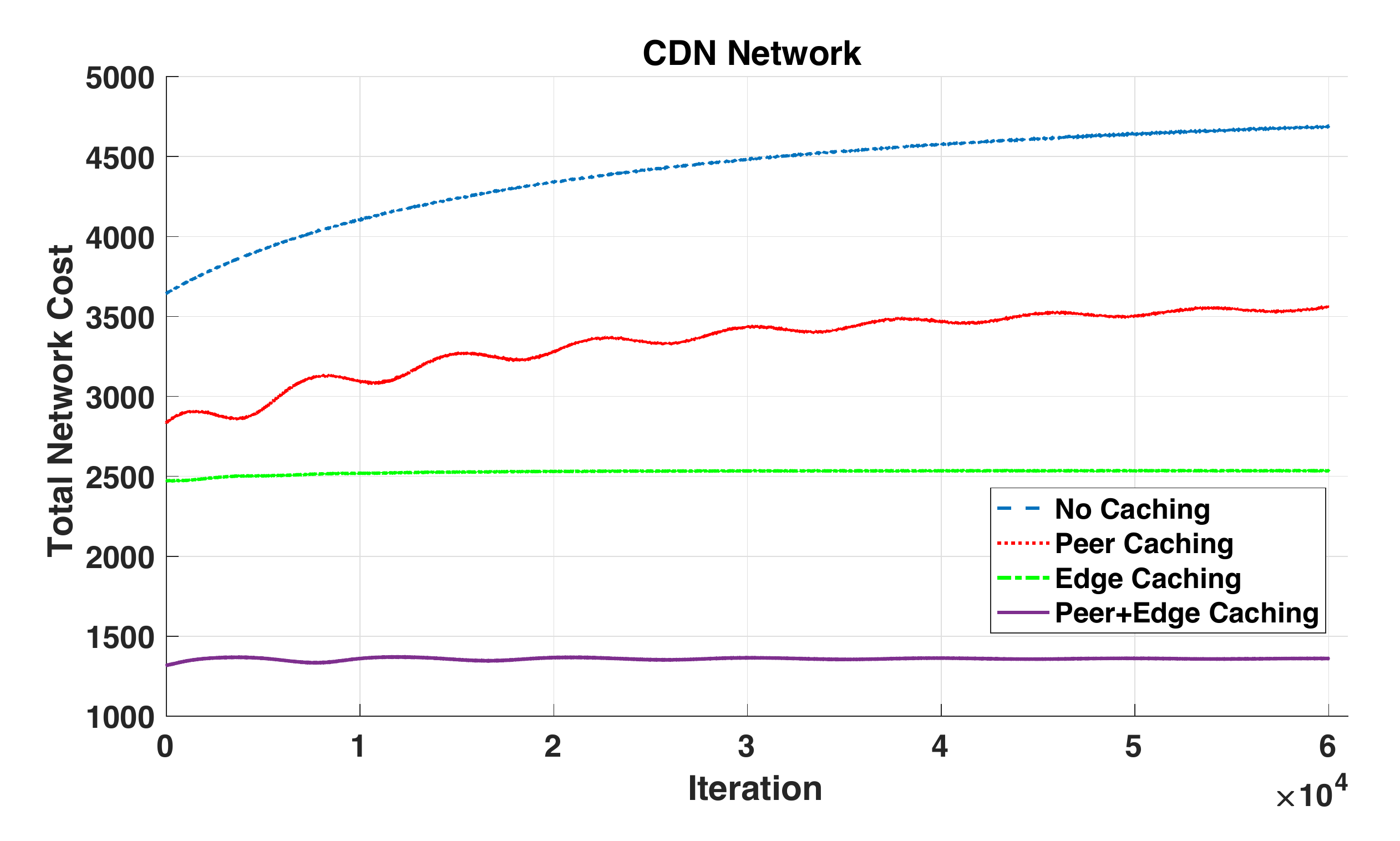}
		\caption{No Caching Cost.}
		\label{fig:cdnn}
	\end{subfigure}%
	\begin{subfigure}[t]{0.33\textwidth}
		\centering
		\includegraphics[width =1\textwidth, height =.8\textwidth ]{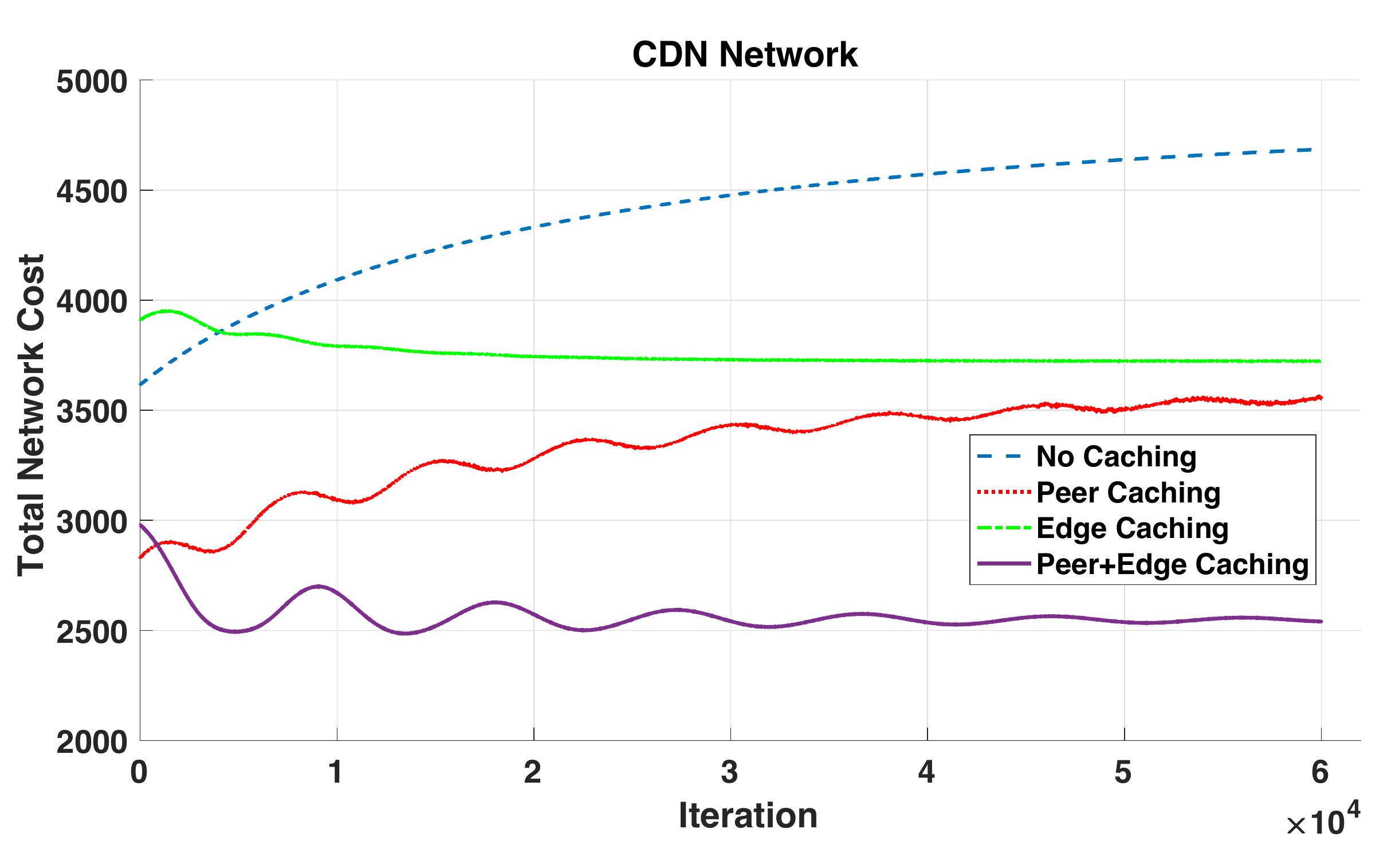}
		\caption{Linear Caching Cost.}
		\label{fig:cdnl}
	\end{subfigure}
	\begin{subfigure}[t]{0.33\textwidth}
		\centering
		\includegraphics[width =1\textwidth, height =.8\textwidth ]{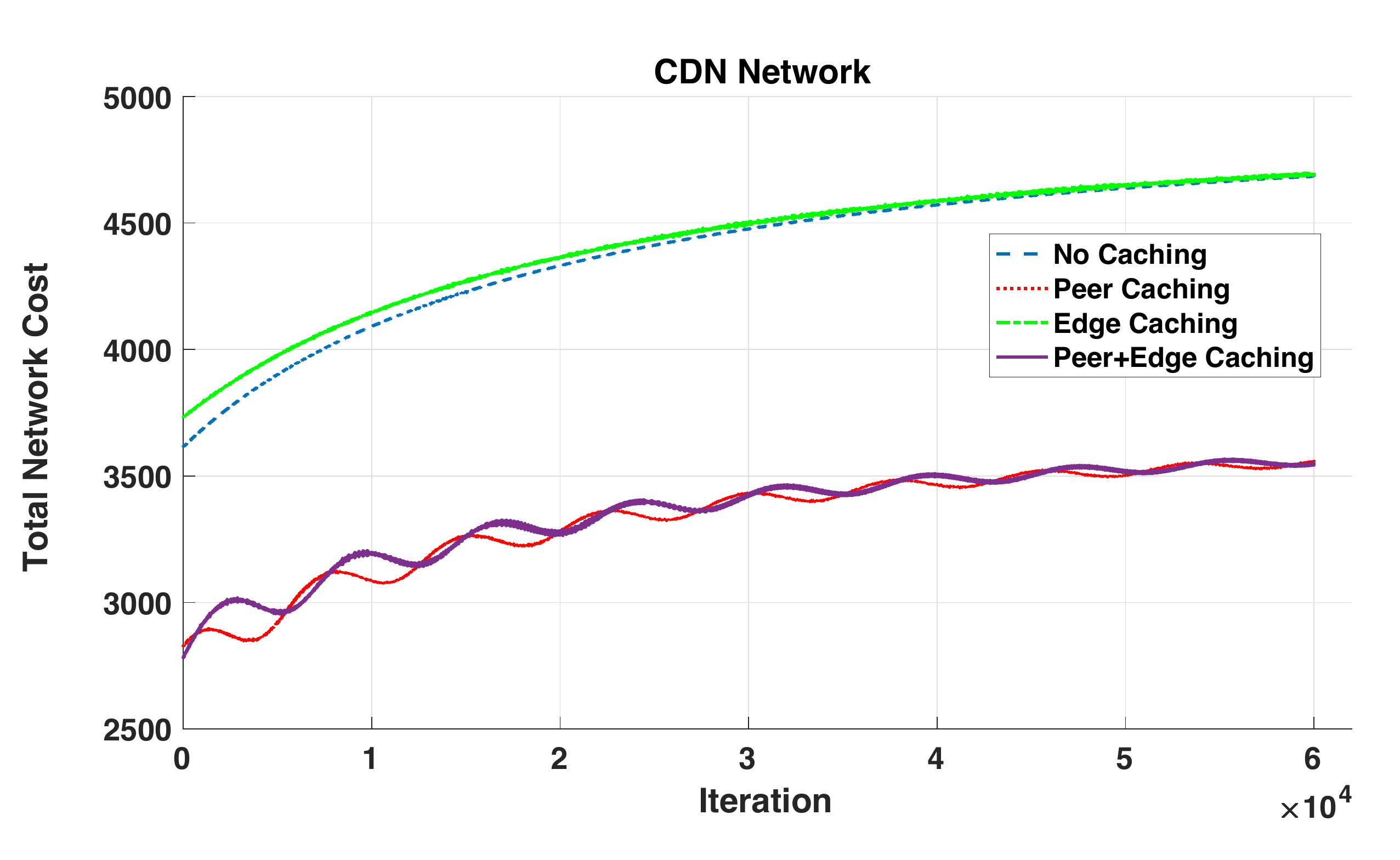}
		\caption{Quadratic Caching Cost.}
		\label{fig:cdnq}
	\end{subfigure}	
	\caption{Convergence of the Primal-Dual Algorithm on the CDN Network under different caching cost functions and scenarios when $B=6$.}
	\label{fig:cdn}
\end{figure*}

\subsection{Cache Placement Experiments}
In the simulation experiments explained in the previous section, it was assumed that the nodes with local storage are fixed. However, choosing the most effective placement of these caches in terms of cost reduction in an arbitrary network topology is not a simple task. We tackle this problem 
by allowing all the nodes in the network, except the source node, to cache and choose the one(s) with highest caching variables. To reduce the randomness, the caching variables are averaged over 40 independent run of the simulation with different random generator seeds. We look at linear and quadratic caching cost functions as examples to study the cache placement problem. Clearly, the case where there is no cost associated with caching functionality will lead to have all caching variables equal to 1. The results are reported in Table \ref{table} for CDN Network and Service Network topologies. 
For the Service Network, as reported in Table \ref{tab:sn}, we note that the average caching variable of node 2 is 0 while, that of node 1 is equal to 1. Hence, it is preferred to cache closer to the source node rather than the edge. This is likely because the number of outgoing link of node 2 is three time greater than that of node 1.  Hence, the required packets to be stored at node 2 is three times larger than that at node 1. For the quadratic caching cost function, we note that the caching variables for both nodes 1 and 2 are zeros. For the both cases, the caching variables for the terminals remain equal to 1. This inertia demonstrates the importance of peer caching in our model. 

For CDN Network, as shown in Table \ref{tab:cdn}, the solution of the caching variables for the edge nodes in the network (nodes 5 and 6) are very close to 1 for linear cost functions, while they are equal to 0 for the quadratic cost functions. Furthermore, it can be seen that the solutions for the caching variables of nodes 1, 2, 3, and 4 are quite similar to each other for both cost functions. This is expected due to the symmetry of the topology. We also note that the solutions for these nodes are all lower for the quadratic case than that for the linear case. This results, along the ones for the Service Network, indicate that depending on the choice of the caching cost function and the topology, we might be better off caching closer to the source vs. at the edge of the network. In particular, in the case of quadratic caching cost functions for the CDN Network, the total network cost is minimized when there are caches at node 1-4.

\begin{table}[ht]
		\caption{Results of Cache Placement Experiment on Service Network and CDN Network: Average Caching Variables for linear and quadratic caching cost functions.}
		\label{table}
		\parbox{.45\linewidth}{
			\centering
			\label{tab:sn} % is used to refer this table in the text
			\begin{tabular}{c c c} % centered columns (4 columns)
				\hline\hline %inserts double horizontal lines
				Node & Linear & Quadratic \\ [0.5ex] % inserts table
				%heading
				\hline % inserts single horizontal line
				1 & 1.00   & 0.00  \\ % inserting body of the table
				2 & 0.00  & 0.00 \\
				$t_1$ & 1.00 &    1.00  \\
				$t_2$ & 1.00 &    1.00  \\
				$t_3$ & 1.00 &    1.00  \\[1ex] % [1ex] adds vertical space
				\hline %inserts single line
			\end{tabular}
			\subcaption{Service Network.} % title of Table
		}
	\parbox{.45\linewidth}{
		\centering
		\label{tab:cdn} % is used to refer this table in the text
	\begin{tabular}{c c c} % centered columns (4 columns)
		\hline\hline %inserts double horizontal lines
		Node & Linear & Quadratic \\ [0.5ex] % inserts table
		%heading
		\hline % inserts single horizontal line
		1 & 0.78   & 0.50  \\ % inserting body of the table
		2 & 0.80  & 0.58 \\
		3 & 0.73  &  0.54  \\
		4 & 0.92  &  0.47 \\
		5 & 1.00    &     0.00 \\
		6 & 0.98    &     0.00  \\
		$t_1$ & 1.00 &    1.00  \\
		$t_2$ & 1.00 &    1.00  \\[1ex] % [1ex] adds vertical space
		\hline %inserts single line
	\end{tabular}
	\subcaption{CDN Network.} % title of Table
}
\end{table}

\section{Conclusion}

{ In summary, we considered a natural generalization of the standard coded-multicast problem~\cite{Lun} where a subset of the nodes have local caches. We were interested in settings where the source had no cache, and where there was a need for frame-by-frame encoding. The setting was motivated by delay-sensitive applications as well as the need to update replicas in distributed object storage systems. We introduced a novel communication scheme for the cache-aided multicast problem, by taking advantage of recent results from the coding theory literature on a problem in function updates~\cite{PraMed}. We first provided natural extensions of  the function updates problem and used these extensions to construct the communication scheme for the cache-aided coded multi-cast problem. The scheme had three important features. Firstly, the scheme is applicable to any general network. Secondly, the scheme works along with any linear network code designed for the same network but without any caches. Finally, the scheme is largely decentralized with respect to encoder design to take advantage of caches. The second and third features are important in practice, and they together handle the temporary nature of caches in the network. For example, in a 1000 round scheme if the cache of a certain node is not available after 500 rounds, the current scheme only requires encoding changes at the in-neighbors of node $i$ without any change else where. Further, the overall linear network code also does not need to change because node $i$'s cache is unavailable.  

Given the link and caching costs, we obtained expressions for the overall cost function for multicast problem. The resultant mixed integer optimization problem was relaxed to a convex problem, whose solution was obtained via primal dual methods. Simulation results were provided for practical settings such as CDN and service networks. In general, when compared to the no caching scenario, the benefits are proportional to the number of nodes having caches. For instance, in the service network, this means that it is beneficial to cache at the destinations (peers) than at the edge. 

Two interesting questions remain at the end of this work, and both of these relate to possibly improving the overall cost function for problem. Firstly,  it is unclear if the function updates coding problem variant in Lemma \ref{lem:func_updates_var2} has a better achievability result than what is presented here. The setting of Lemma \ref{lem:func_updates_var2} is MAC variation of the problem in \cite{PraMed}, and this by itself is an interesting coding theory problem, especially  given the importance in network context as identified in this work. Secondly, it is interesting to explore if the linear network code (LNC) itself can be obtained to take advantage of the caches. In the current work, we first fix the the global LNC, and then use local encoding technique (at the in-neighbors) to take advantage of cache.  While this approach provided practical advantages as mentioned above, it might be possible to reduce the overall cost by considering a joint design of the LNC and caching scheme. In this regard, it is an interesting exercise to explore if algebraic framework in \cite{KoetterMedard} can be suitably modified to directly handle the presence of any cache in the network.
}

\bibliographystyle{IEEEtran}
\bibliography{citations}

\begin{thebibliography}{10}
\providecommand{\url}[1]{#1}
\csname url@samestyle\endcsname
\providecommand{\newblock}{\relax}
\providecommand{\bibinfo}[2]{#2}
\providecommand{\BIBentrySTDinterwordspacing}{\spaceskip=0pt\relax}
\providecommand{\BIBentryALTinterwordstretchfactor}{4}
\providecommand{\BIBentryALTinterwordspacing}{\spaceskip=\fontdimen2\font plus
\BIBentryALTinterwordstretchfactor\fontdimen3\font minus
  \fontdimen4\font\relax}
\providecommand{\BIBforeignlanguage}[2]{{%
\expandafter\ifx\csname l@#1\endcsname\relax
\typeout{** WARNING: IEEEtran.bst: No hyphenation pattern has been}%
\typeout{** loaded for the language `#1'. Using the pattern for}%
\typeout{** the default language instead.}%
\else
\language=\csname l@#1\endcsname
\fi
#2}}
\providecommand{\BIBdecl}{\relax}
\BIBdecl

\bibitem{Lun}
D.~S. Lun, N.~Ratnakar, M.~Medard, R.~Koetter, D.~R. Karger, T.~Ho, E.~Ahmed,
  and F.~Zhao, ``Minimum-cost multicast over coded packet networks,''
  \emph{IEEE Transactions on Information Theory}, vol.~52, no.~6, pp.
  2608--2623, June 2006.

\bibitem{chen2017delay}
Y.-H. Chen, C.-C. Hu, E.~H.-K. Wu, S.-M. Chuang, and G.-H. Chen, ``A
  delay-sensitive multicast protocol for network capacity enhancement in
  multirate manets,'' \emph{IEEE Systems Journal}, 2017.

\bibitem{chou2003practical}
P.~A. Chou, Y.~Wu, and K.~Jain, ``Practical network coding,'' in
  \emph{Proceedings of the annual Allerton conference on communication control
  and computing}, vol.~41, no.~1.\hskip 1em plus 0.5em minus 0.4em\relax The
  University; 1998, 2003, pp. 40--49.

\bibitem{gkantsidis2005network}
C.~Gkantsidis and P.~R. Rodriguez, ``Network coding for large scale content
  distribution,'' in \emph{INFOCOM 2005. 24th Annual Joint Conference of the
  IEEE Computer and Communications Societies. Proceedings IEEE}, vol.~4.\hskip
  1em plus 0.5em minus 0.4em\relax IEEE, 2005, pp. 2235--2245.

\bibitem{llorca2013network}
J.~Llorca, A.~M. Tulino, K.~Guan, and D.~C. Kilper, ``Network-coded
  caching-aided multicast for efficient content delivery,'' in
  \emph{Communications (ICC), 2013 IEEE International Conference on}.\hskip 1em
  plus 0.5em minus 0.4em\relax IEEE, 2013, pp. 3557--3562.

\bibitem{fragouli2007network}
C.~Fragouli, E.~Soljanin \emph{et~al.}, ``Network coding fundamentals,''
  \emph{Foundations and Trends{\textregistered} in Networking}, vol.~2, no.~1,
  pp. 1--133, 2007.

\bibitem{li2003linear}
S.-Y. Li, R.~W. Yeung, and N.~Cai, ``Linear network coding,'' \emph{IEEE
  transactions on information theory}, vol.~49, no.~2, pp. 371--381, 2003.

\bibitem{rlnc}
T.~Ho, M.~Medard, R.~Koetter, D.~R. Karger, M.~Effros, J.~Shi, and B.~Leong,
  ``A random linear network coding approach to multicast,'' \emph{IEEE
  Transactions on Information Theory}, vol.~52, no.~10, pp. 4413--4430, Oct
  2006.

\bibitem{maddah2014fundamental}
M.~A. Maddah-Ali and U.~Niesen, ``Fundamental limits of caching,'' \emph{IEEE
  Transactions on Information Theory}, vol.~60, no.~5, pp. 2856--2867, 2014.

\bibitem{maddah2015decentralized}
------, ``Decentralized coded caching attains order-optimal memory-rate
  tradeoff,'' \emph{IEEE/ACM Transactions On Networking}, vol.~23, no.~4, pp.
  1029--1040, 2015.

\bibitem{niesen2017coded}
U.~Niesen and M.~A. Maddah-Ali, ``Coded caching with nonuniform demands,''
  \emph{IEEE Transactions on Information Theory}, vol.~63, no.~2, pp.
  1146--1158, 2017.

\bibitem{park2017coded}
S.-H. Park, O.~Simeone, W.~Lee, and S.~Shamai, ``Coded multicast fronthauling
  and edge caching for multi-connectivity transmission in fog radio access
  networks,'' \emph{arXiv preprint arXiv:1705.04070}, 2017.

\bibitem{hassanzadeh2016cache}
P.~Hassanzadeh, A.~Tulino, J.~Llorca, and E.~Erkip, ``Cache-aided coded
  multicast for correlated sources,'' in \emph{Turbo Codes and Iterative
  Information Processing (ISTC), 2016 9th International Symposium on}.\hskip
  1em plus 0.5em minus 0.4em\relax IEEE, 2016, pp. 360--364.

\bibitem{yang2017centralized}
Q.~Yang and D.~G{\"u}nd{\"u}z, ``Centralized coded caching of correlated
  contents,'' \emph{arXiv preprint arXiv:1711.03798}, 2017.

\bibitem{PraMed}
\BIBentryALTinterwordspacing
N.~Prakash and M.~M{\'{e}}dard, ``Communication cost for updating linear
  functions when message updates are sparse: Connections to maximally
  recoverable codes,'' \emph{CoRR}, vol. abs/1605.01105, 2016. [Online].
  Available: \url{http://arxiv.org/abs/1605.01105}
\BIBentrySTDinterwordspacing

\bibitem{cheng2017delay}
T.~Y. Cheng and X.~Jia, ``Delay-sensitive multicast in inter-datacenter wan
  using compressive latency monitoring,'' \emph{IEEE Transactions on Cloud
  Computing}, 2017.

\bibitem{calinescu2011}
G.~Calinescu, C.~Chekuri, M.~P{\'a}l, and J.~Vondr{\'a}k, ``Maximizing a
  monotone submodular function subject to a matroid constraint,'' \emph{SIAM
  Journal on Computing}, vol.~40, no.~6, pp. 1740--1766, 2011.

\bibitem{sigmetrics}
S.~Ioannidis and E.~Yeh, ``Adaptive caching networks with optimality
  guarantees,'' in \emph{Proceedings of the 2016 ACM SIGMETRICS International
  Conference on Measurement and Modeling of Computer Science}.\hskip 1em plus
  0.5em minus 0.4em\relax ACM, 2016, pp. 113--124.

\bibitem{KoetterMedard}
R.~Koetter and M.~Medard, ``An algebraic approach to network coding,''
  \emph{IEEE/ACM Transactions on Networking}, vol.~11, no.~5, pp. 782--795, Oct
  2003.

\bibitem{internetCC}
R.~Srikant, \emph{The mathematics of Internet congestion control}.\hskip 1em
  plus 0.5em minus 0.4em\relax Springer Science \& Business Media, 2012.

\end{thebibliography}
\appendices
\section{Proof of Theorem \ref{theorem:stability}}
\label{app:convergence}
	To prove the convergence of primal-dual algorithm to a globally optimal solution of the relaxed problem, we use Lyapunov stability theory, and show that the proposed algorithm is globally, asymptotically stable. This proof is based on the proof of Theorem 3.7 of \cite{internetCC}, and Proposition 1 of \cite{Lun}.
	
	Following the equilibrium point $(\widehat{\Sigma},\widehat{K},\widehat{P},\widehat{\Lambda},\widehat{\Gamma})$ satisfying KKT conditions in (\ref{eq:kkt-lag_mu})-(\ref{eq:kkt_deltagamma}), we consider the function given in (\ref{eq:V}) as a candidate for the Lyapunov function. 
	Note that $V(\widehat{\Sigma},\widehat{K},\widehat{P},\widehat{\Lambda},\widehat{\Gamma})=0$. Since $k_{ij}^{(t)}(x)>0$, if $\mu^{(t)}_{i,j}\neq \widehat{\mu}^{(t)}_{i,j}$, we have
	
	$$\int_{\widehat{\mu}^{(t)}_{i,j}}^{\mu^{(t)}_{i,j}} \frac{1}{k_{ij}^{(t)}(x)} (x - \widehat{\mu}^{(t)}_{i,j})d \kappa>0.$$
	
	Similarly, we can extend this argument to the other terms in $V$, hence, we have $V(\Sigma,K,P,\Lambda,\Gamma)>0$ whenever $(\Sigma,K,P,\Lambda,\Gamma)\neq (\widehat{\Sigma},\widehat{K},\widehat{P},\widehat{\Lambda},\widehat{\Gamma})$.
	
	Proceeding to the $\dot{V}$ in (\ref{eq:Vdot}) , let us first prove the following.
	
	\begin{equation}
	\left(-\mu^{(t)}_{i,j}\right)_{\lambda^{(t)}_{i,j}}^+(\lambda^{(t)}_{i,j} - \widehat{\lambda}^{(t)}_{i,j})
	\leq -\mu^{(t)}_{i,j}(\lambda^{(t)}_{i,j} - \widehat{\lambda}^{(t)}_{i,j}). 
	\label{eq:ineq-ldot}
	\end{equation}
	The above inequality is an equality if either $\mu^{(t)}_{i,j}\leq 0$ or $\lambda^{(t)}_{i,j}>0$. On the other hand, when $\mu^{(t)}_{i,j}>0$ and $\lambda^{(t)}_{i,j}\leq 0$, we have $\left(-\mu^{(t)}_{i,j}\right)_{\lambda^{(t)}_{i,j}}^+=0$, and since $\widehat{\lambda}^{(t)}_{i,j}\geq 0$, it follows $-\mu^{(t)}_{i,j}(\lambda^{(t)}_{i,j} - \widehat{\lambda}^{(t)}_{i,j})\geq 0$. Hence, (\ref{eq:ineq-ldot}) holds. Similarly, it can be verified that
	$$\left(-\kappa_i\right)_{\gamma_i^-}^+(\gamma_i^--\widehat{\gamma}_i^-)\leq
	-\kappa_i(\gamma_i^--\widehat{\gamma}_i^-),$$
	$$\left(\kappa_i-1\right)_{\gamma_i^+}^+(\gamma_i^+-\widehat{\gamma}_i^+)
	\leq	(\kappa_i-1)(\gamma_i^+-\widehat{\gamma}_i^+).$$

	 Thus, we get the (a) inequality. By applying KKT conditions (\ref{eq:kkt-lag_mu})-(\ref{eq:kkt_deltagamma}) and noting that

	\begin{IEEEeqnarray}{+rCl+x*}
		p'y &=& \sum\limits_{t\in T}\sum\limits_{i\in\mathcal{N}} p_i^{(t)}\left(\sum\limits_{\{j\in \mathcal{N}_{out}(i)\}}  \mu^{(t)}_{i,j} - \sum\limits_{\{j \in \mathcal{N}_{in}(i)\}}  \mu^{(t)}_{j,i} \right)\nonumber\\
		&=& \sum\limits_{t\in T}\sum\limits_{(i,j)\in \mathcal{E}}
		\mu^{(t)}_{i,j}(p_i^{(t)}-p_j^{(t)}) = q'\mu,
	\end{IEEEeqnarray}
	and further re-arrangement of the terms, equation (b) and then (c) follows. Since $\Psi(\Sigma,K)$ is strictly convex in $\Sigma$, and linear in $K$, $$\left(\triangledown_{\Sigma} \Phi(\widehat{\Sigma},\widehat{K})- \triangledown_{\Sigma} \Psi(\Sigma,K)\right)' (\Sigma-\widehat{\Sigma}) \leq 0,$$
	 $$\left(\triangledown_{K} \Phi(\widehat{\Sigma},\widehat{K})- \triangledown_{K} \Psi(\Sigma,K)\right)' (K-\widehat{K})= 0,$$ 
	 and thus, $\dot{V}\leq - \Lambda' \widehat{\Sigma} - (\mathbf{1}-\widehat{K})'\Gamma^+ - (\widehat{K})'\Gamma^-$, with equality if and only if $\Sigma=\widehat{\Sigma}$. 
	
	Note that, if the initial choice of $\Lambda$, $\Gamma^-$ and $\Gamma^+$ are element-wise non-negative, that is, $\Lambda(0),\Gamma^-(0),\Gamma^+(0)\succeq 0$, it can be verified from the primal-dual algorithm, that $\Lambda(n),\Gamma^-(n),\Gamma^+(n)\succeq 0 $, where $n =0,1,2,\ldots$ represents the algorithm iteration. Assuming $\Lambda,\Gamma^-,\Gamma^+\succeq 0$, it follows that $\dot{V}\leq 0$ since $\widehat{\Sigma}\succeq 0$ and $0\preceq \widehat{K}\preceq 1$. Therefore, the primal-dual algorithm is globally, asymptotically stable, and hence, it converges to a globally optimal solution of the relaxed problem.
	
	\begin{figure*}
		\rule{\textwidth}{1pt}
		\begin{eqnarray}
		&& V(\Sigma,K,P,\Lambda,\Gamma) = \sum\limits_{ t\in T}\biggl\lbrace \sum\limits_{(i,j)\in \mathcal{E}}\left(
		\int_{\widehat{\mu}^{(t)}_{i,j}}^{\mu^{(t)}_{i,j}} \frac{1}{k_{ij}^{(t)}(x)} (x - \widehat{\mu}^{(t)}_{i,j})d x + 
		\int_{\widehat{\lambda}^{(t)}_{ij}}^{\lambda^{(t)}_{i,j}} \frac{1}{m_{i,j}^{(t)}(s)} (s - \widehat{\lambda}^{(t)}_{i,j})d s\right)+\label{eq:V}\\&&
		\sum\limits_{i\in\mathcal{N}}\int_{\widehat{p}^{(t)}_{i}}^{p^{(t)}_{i}} \frac{1}{g_{i}^{(t)}(z)} (z - \widehat{p}^{(t)}_{i})d z\vphantom{\sum\limits_{(i,j)\in \mathcal{E}}}\biggr\rbrace +
		\sum\limits_{i\in\mathcal{N}} \biggl\lbrace  \int_{\widehat{\kappa}_{i}}^{\kappa_{i}} \frac{1}{h_{i}(l)} (l - \widehat{\kappa}_{i})d l+
		 \int_{\widehat{\gamma}^-_{i}}^{\gamma^-_{i}} \frac{1}{\alpha_{i}(v)} (v - \widehat{\gamma}^-_{i})d v+
		 \int_{\widehat{\gamma}^+_{i}}^{\gamma^+_{i}} \frac{1}{\beta_{i}(w)} (w - \widehat{\gamma}^+_{i})d w \biggr\rbrace.\nonumber
		\end{eqnarray}
		\rule{\textwidth}{1pt}
		\begin{eqnarray}
		&\dot{V}& = \sum\limits_{ t\in T}\biggl\lbrace \sum\limits_{(i,j)\in \mathcal{E}}\left(
		\left(-\frac{\partial \Psi(\Sigma,K)}{\partial \mu^{(t)}_{i,j}} - q^{(t)}_{i,j}+ \lambda^{(t)}_{i,j}\right)\cdot (\mu^{(t)}_{i,j} - \widehat{\mu}^{(t)}_{i,j}) + 
		\left(-\mu^{(t)}_{i,j}\right)_{\lambda^{(t)}_{i,j}}^+(\lambda^{(t)}_{i,j} - \widehat{\lambda}^{(t)}_{i,j})\right)\nonumber \\ &+&
		\sum\limits_{i\in\mathcal{N}}\left(y^{(t)}_{i}-\theta^{(t)}_{i}\right) (p^{(t)}_{i} - \widehat{p}^{(t)}_{i})\vphantom{\sum\limits_{(i,j)\in \mathcal{E}}}\biggr\rbrace+
		\sum\limits_{i\in\mathcal{N}}\biggl\lbrace \left(-\frac{\partial \Psi(\Sigma,K)}{\partial \kappa_i} -\gamma_i^++\gamma_i^-\right)(\kappa_i - \widehat{\kappa}_i)\nonumber\\&+&
		\left(-\kappa_i\right)_{\gamma_i^-}^+(\gamma_i^--\widehat{\gamma}_i^-)+
		\left(\kappa_i-1\right)_{\gamma_i^+}^+(\gamma_i^+-\widehat{\gamma}_i^+)
		\biggr\rbrace.
		\label{eq:Vdot}
		\end{eqnarray}
		\rule{\textwidth}{1pt}
		\begin{eqnarray*}
			\dot{V} &\stackrel{(a)}{\leq} & \sum\limits_{ t\in T}\biggl\lbrace \sum\limits_{(i,j)\in \mathcal{E}}\left(
			\left(-\frac{\partial \Psi(\Sigma,K)}{\partial \mu^{(t)}_{i,j}} - q^{(t)}_{i,j}+ \lambda^{(t)}_{i,j}\right)\cdot (\mu^{(t)}_{i,j} - \widehat{\mu}^{(t)}_{i,j}) -\mu^{(t)}_{i,j}(\lambda^{(t)}_{i,j} - \widehat{\lambda}^{(t)}_{i,j})\right)\nonumber \\ &+&
			\sum\limits_{i\in\mathcal{N}}\left(y^{(t)}_{i}-\theta^{(t)}_{i}\right) (p^{(t)}_{i} - \widehat{p}^{(t)}_{i})\vphantom{\sum\limits_{(i,j)\in \mathcal{E}}}\biggr\rbrace+
			\sum\limits_{i\in\mathcal{N}}\biggl\lbrace \left(-\frac{\partial \Psi(\Sigma,K)}{\partial \kappa_i} -\gamma_i^++\gamma_i^-\right)(\kappa_i - \widehat{\kappa}_i)\nonumber\\&-&
			\kappa_i(\gamma_i^--\widehat{\gamma}_i^-)+
			(\kappa_i-1)(\gamma_i^+-\widehat{\gamma}_i^+)
			\biggr\rbrace
			  \stackrel{(b)}{=} (\widehat{q}-q)'(\mu-\widehat{\mu})-(\widehat{p}-p)'(y-\widehat{y})\\
			&+& \sum\limits_{ t\in T}\biggl\lbrace \sum\limits_{(i,j)\in \mathcal{E}}\left(
			\left(\frac{\partial \Psi(\widehat{\Sigma},\widehat{K})}{\partial \mu^{(t)}_{i,j}}-\frac{\partial \Psi(\Sigma,K)}{\partial \mu^{(t)}_{i,j}}\right) (\mu^{(t)}_{i,j} - \widehat{\mu}^{(t)}_{i,j}) -\widehat{\mu}^{(t)}_{i,j}(\lambda^{(t)}_{i,j}-\widehat{\lambda}^{(t)}_{i,j})\right)\biggr\rbrace
			\\&+&
			\sum\limits_{i\in\mathcal{N}}\biggl\lbrace \left(\frac{\partial \Psi(\widehat{\Sigma},\widehat{K})}{\partial \kappa_i}-\frac{\partial \Psi(\Sigma,K)}{\partial \kappa_i} \right)(\kappa_i - \widehat{\kappa}_i)-
			\widehat{\kappa}_i(\gamma_i^--\widehat{\gamma}_i^-)+
			(1-\widehat{\kappa}_i)(\widehat{\gamma}_i^+ - \gamma_i^+)
			\biggr\rbrace\\
			&\stackrel{(c)}{=}&
			 \left(\triangledown_{\Sigma} \Phi(\widehat{\Sigma},\widehat{K})- \triangledown_{\Sigma} \Psi(\Sigma,K)\right)' (\Sigma-\widehat{\Sigma})+ \left(\triangledown_{K} \Phi(\widehat{\Sigma},\widehat{K})- \triangledown_{K} \Psi(\Sigma,K)\right)' (K-\widehat{K}) \\&-& \Lambda' \widehat{\Sigma} - (\mathbf{1}-\widehat{K})'\Gamma^+ - (\widehat{K})'\Gamma^- \stackrel{(d)}{\leq} - \Lambda' \widehat{\Sigma} - (\mathbf{1}-\widehat{K})'\Gamma^+ - (\widehat{K})'\Gamma^-.
			\label{eq:Vdot2} 
		\end{eqnarray*}
	\end{figure*}
	
\end{document}